\documentclass{llncs}[11pt]
\usepackage[letterpaper,hmargin=1in,vmargin=1in]{geometry}

\usepackage{algorithmicx}
\usepackage{algpseudocode}
\usepackage{algorithm}
\usepackage{float}

\usepackage{amsmath}
\usepackage{amsfonts}
\usepackage{amssymb}
\usepackage{url}
\usepackage[stable]{footmisc}
\usepackage{cite}
\DeclareMathAlphabet{\mathcalligra}{T1}{calligra}{m}{n}
\numberwithin{equation}{section}

\floatstyle{boxed}

\renewcommand{\vec}{\mathbf}

\newcommand{\ZZ}{\mathbb{Z}}

\newcommand{\QR}{\mathbb{QR}}

\newcommand{\fname}{\mathsf}
\newcommand{\aname}[1]{\mathbf{\fname{#1}}}

\newcommand{\vname}{\mathsf}

\newcommand{\ind}{\approx}
\newcommand{\cind}{\underset{C}{\ind}}
\newcommand{\defeq}{\triangleq}
\newcommand{\samplefrom}{\xleftarrow{\$}}
\newcommand{\adv}[2]{\mathbf{Adv}_{#1}^{#2}}
\newcommand{\exper}[2]{\mathbf{Exp}_{#1}^{#2}}
\newcommand{\pr}{\fname{Pr}}

\renewcommand{\implies}{\Rightarrow}

\makeatletter
\def\legendre@dash#1#2{\hb@xt@#1{%
  \kern-#2\p@
  \cleaders\hbox{\kern.5\p@
    \vrule\@height.2\p@\@depth.2\p@\@width\p@
    \kern.5\p@}\hfil
  \kern-#2\p@
  }}
\def\@legendre#1#2#3#4#5{\mathopen{}\left(
  \sbox\z@{$\genfrac{}{}{0pt}{#1}{#3#4}{#3#5}$}%
  \dimen@=\wd\z@
  \kern-\p@\vcenter{\box0}\kern-\dimen@\vcenter{\legendre@dash\dimen@{#2}}\kern-\p@
  \right)\mathclose{}}
\newcommand\legendre[2]{\mathchoice
  {\@legendre{0}{1}{}{#1}{#2}}
  {\@legendre{1}{.5}{\vphantom{1}}{#1}{#2}}
  {\@legendre{2}{0}{\vphantom{1}}{#1}{#2}}
  {\@legendre{3}{0}{\vphantom{1}}{#1}{#2}}
}
\def\dlegendre{\@legendre{0}{1}{}}
\def\tlegendre{\@legendre{1}{0.5}{\vphantom{1}}}
\makeatother

\newcommand{\jacobi}[2]{\genfrac(){}{0}{#1}{#2}}

\usepackage{alltt}

\begin{document}
%
\title{Homomorphic Encryption with Access Policies: Characterization
and New Constructions\footnote{A preliminary version of this work appeared in Africacrypt 2013 \cite{CHT2013af}. This is the full version.}}

\author{Michael Clear\footnotemark[1], Arthur Hughes, and Hitesh Tewari}
\institute{School of Computer Science and Statistics,\\Trinity College
Dublin}

\footnotetext[1]{The author's work is funded by the Irish Research Council EMBARK Initiative.}

\maketitle

\begin{abstract}
A characterization of predicate encryption (PE) with support for
homomorphic operations is presented and we describe the homomorphic properties of some
existing PE constructions. Even for the special case of IBE, there are few
known group-homomorphic cryptosystems.
Our main construction is an XOR-homomorphic IBE scheme based on the
quadratic residuosity problem (variant of the Cocks' scheme), which we
show to be strongly homomorphic. We were unable to construct an anonymous variant that preserves this
homomorphic property, but we achieved anonymity for a weaker notion of homomorphic
encryption, which we call \emph{non-universal}.  A related security
notion for this weaker primitive is formalized. Finally, some
potential applications
and open problems are considered.
\end{abstract}

\section{Introduction}
There has been much interest recently in encryption schemes with
homomorphic capabilities. Traditionally, malleability was avoided to
satisfy strong security definitions, but many applications have been
identified for cryptosystems supporting homomorphic operations. More
recently, Gentry \cite{Gentry2009} presented the first fully-homomorphic
encryption (FHE) scheme, and several improvements and variants have since appeared in
the literature \cite{Smart10,Dijk10,BV11a,Brakerski2011b}. There are however many
applications that only require a scheme to support a single
homomorphic operation. Such schemes are referred to as \emph{partial
  homomorphic}. Notable examples of unbounded homomorphic cryptosystems include
Goldwasser-Micali \cite{GM82} (XOR), Paillier
\cite{Paillier:1999:PKC}  and ElGamal
\cite{ElGamal:1985:PKCa}.

Predicate Encryption (PE) \cite{Katz:2008} enables a sender to embed a hidden
descriptor
within a ciphertext that consists of attributes describing the message
content. A 
Trusted Authority (TA) who manages the system issues secret keys to users
corresponding to predicates. A user can decrypt a ciphertext
containing a descriptor $\vec{a}$ if and only if he/she has a secret
key for a predicate that evaluates to true for
$\vec{a}$. This construct turns out to be quite powerful, and
generalizes many encryption primitives. It facilitates
expressive fine-grained access control i.e. complex policies can be defined
restricting the recipients who can decrypt a message. It also facilitates
the evaluation of complex queries on data such as range, subset
and search queries. Extending the class of supported predicates for known
schemes is a topic of active research at present. 

PE can be
viewed in two ways. It can be viewed as a means to delegate
computation to a third party i.e. allow the third party to perform a
precise fixed function on the encrypted data, and thus limit what the third
party learns about the data. In the spirit of this viewpoint, a
generalization known as Functional Encryption has been proposed
\cite{Boneh11}, which allows general
functions to be evaluated.

PE can also be viewed as a means to achieve more fine-grained access
control. It enables a stronger separation between sender and recipient 
since the former must only describe the content of the message or more
general conditions
on its access while decryption then depends on whether a recipient's access
policy matches these conditions.

Why consider homomorphic encryption in the PE setting? It is
conceivable that in a multi-user environment such as a large organization, certain computations may be delegated to the cloud whose
inputs depend on the work of multiple users distributed within that organization. Depending on
the application, the circuit to be computed may be chosen or adapted
by the cloud provider, and thus is not fixed by the delegator as in
primitives such as non-interactive
\emph{verifiable computing} \cite{Gennaro:2010}. Furthermore, the
computation may depend on data sets provided by multiple independent
users. Since the data is potentially sensitive, the organization's
security policy may dictate that all data must be encrypted. Accordingly,
each user encrypts her data with a PE scheme using
relevant attributes to describe it. She then sends the ciphertext(s)
to the cloud. It is desirable that the results of the computation
returned from the cloud be decryptable only by an entity whose access
policy (predicate) satisfies the attributes of \emph{all} data sets
used in the computation. Of course a public-key homomorphic scheme
together with a PE scheme would be sufficient if the senders were able to
interact before contacting the cloud, but we would like to remove this
requirement since the senders may not be aware of each other. This
brings to mind the recent notion of multikey homomorphic
encryption presented by L\'{o}pez-Alt, Tromer and Vaikuntanathan
\cite{Lopez-Alt:2012}.

Using a multikey homomorphic scheme, the
senders need not interact with each other before evaluation takes
place on the cloud. Instead, they must run an MPC decryption protocol to jointly decrypt the
result produced by the cloud. The evaluated ciphertexts in the scheme
described in
\cite{Lopez-Alt:2012} do not depend on the circuit size, and depend
only polynomially on the security parameter and the number of parties who contribute inputs to the circuit. Therefore, the problem
outlined above may be solved with a multikey fully homomorphic scheme
used in conjunction with a PE scheme if we accept the evaluated
ciphertext size to be polynomial in the number of parties. In this
work, we are concerned with a ciphertext size that is independent of
the number of parties. Naturally, this limits the composition of
access policies, but if this is acceptable in an application, there
may be efficiency gains over the combination of multikey FHE and PE.

In summary, homomorphic encryption in the PE setting is desirable if
there is the possibility of multiple parties in a large organization (say) sending encrypted data to
a semi-trusted
\footnote{We assume all parties are semi-honest.} evaluator and access policies are required to
appropriately limit access to the results, where the ``composition''
of access policies is ``lossy''. We assume the semi-honest model in
this paper; in particular we do not consider verifiability of the
computation.

The state of affairs for homomorphic encryption even for the
simplest special case of PE, namely identity-based encryption (IBE),
leaves open many challenges. At his talk at Crypto 2010, Naccache
\cite{NaccacheTalk10} mentioned
``identity-based fully homomorphic encryption'' as one of a list of theory
questions. Towards this goal, it has been pointed out in \cite{Brakerski:2011_v1} that some
LWE-based FHE constructions can be modified to obtain a weak form of an
identity-based FHE scheme using the trapdoor functions from
\cite{GPV}; that is, additional information is
needed (beyond what can be non-interactively derived from a user's
identity) in order to evaluate certain circuits and to perform bootstrapping. Therefore, the valued
non-interactivity property of IBE is lost whereby no communication
between encryptors and the TA is needed. To the best of our knowledge,
fully-homomorphic or even ``somewhat-homomorphic'' IBE remains
open, and a variant of the BGN-type scheme of Gentry, Halevi and
Vaikuntanathan \cite{Gentry:2010} is the only IBE scheme that can compactly
evaluate quadratic formulae (supports 2-DNF).

As far as the authors are aware, there are no $(\ZZ_N, +)$ (like Paillier) or
$(\ZZ_p^\ast, \ast)$ (like ElGamal) homomorphic IBE schemes. Many
pairings-based IBE constructions admit multiplicative homomorphisms which give us a limited additive homomorphism for small ranges;
that is, a discrete logarithm problem must be solved to recover the
plaintext, and the complexity thereof is $O(\sqrt{M})$, where $M$ is
the size of the message space. Of a similar variety are public-key schemes
such as BGN \cite{BGN05} and Benaloh \cite{Benaloh:1994}. It remains
open to construct an unbounded additively homomorphic IBE scheme for a
``large'' range such as Paillier \cite{Paillier:1999:PKC}. Possibly a
fruitful step in this direction would be to look at Galbraith's
variant of Paillier's cryptosystem based on elliptic curves over rings \cite{Galbraith:2002}.

One of the contributions of this paper is to construct an additively
homomorphic IBE scheme for $\ZZ_2$, which is usually referred to as
XOR-homomorphic. XOR-homomorphic schemes such as Goldwasser-Micali
\cite{GM82} have been used in many practical applications including
sealed-bid auctions, biometric authentication and as the building
blocks of protocols such as private information retrieval, and it
seems that an IBE XOR-homomorphic scheme may be useful in some of
these scenarios.

We faced barriers however trying to make our XOR-homomorphic scheme
anonymous. The main obstacle is that the homomorphism depends on the
public key.  We pose as an open problem the task of constructing a
variant that achieves anonymity and retains the homomorphic property. Inheriting the terminology of Golle et al. \cite{Golle:2002} (who
refer to re-encryption without the public key as \emph{universal}
re-encryption), we designate homomorphic evaluation in a
scheme that does not require knowledge of the public key
as \emph{universal}. We introduce a weaker primitive that explicitly
requires additional information to be passed to the homomorphic
evaluation algorithm. Our construction can be made anonymous and
retain its homomorphic property in this context; that is, if the attribute (identity in the case
of IBE) is known to an
evaluator. While this certainly is not ideal, it may be
plausible in some scenarios that an evaluator is allowed to be privy to the
attribute(s) encrypted by the
ciphertexts, and it is other parties in the system to whom the
attribute(s) must remain concealed.  An adversary sees incoming and
outgoing ciphertexts, and can potentially request evaluations on
arbitrary ciphertexts. We call such a variant \emph{non-universal}.
We propose a syntax for a non-universal homomorphic primitive and
formulate a security notion to capture attribute-privacy in this
context.

\subsection{Related Work}
There have been several endeavors to characterize homomorphic
encryption schemes. Gj\o steen \cite{Gjosteen:2005} succeeded in
characterizing many well-known group homomorphic cryptosystems by
means of an abstract construction whose security rests on the hardness
of a subgroup membership problem. More recently, Armknecht,
Katzenbeisser and Peter \cite{Armknecht:2012} gave a more complete
characterization and generalized Gj\o steen's results
to the IND-CCA1 setting. However, in this work, our focus is at a higher level
and not concerned with the underlying algebraic structures. In
particular, we do not require the homomorphisms to be
unbounded since our aim to provide a more general
characterization for homomorphic encryption in the PE setting. Compactness, however, is required; that is, informally, the length of
an evaluated ciphertext should be independent of the \emph{size} of
the computation.

The notion of receiver-anonymity or key-privacy was
formally established by Bellare et al. \cite{Bellare01}, and the
concept of universal anonymity (any user can anonymize a ciphertext)
was proposed in \cite{Hayashi:2005}. The first universally
anonymous IBE scheme appeared in \cite{AG09}. Prabhakaran and Rosulek \cite{Prabhakaran:2008} consider receiver-anonymity for their definitions
of homomorphic encryption.

Finally, since Cocks' IBE scheme \cite{COCKS01} appeared, variants have been proposed
(\cite{BGH07} and \cite{AG09}) that achieve anonymity and improve
space efficiency. However, the possibility of constructing a homomorphic
variant has not received
attention to date.

\subsection{Organization}
Notation and background definitions are set out in Section
\ref{sec:prelim}. Our characterization of homomorphic predicate
encryption is specified in Section \ref{sec:hpe}; the syntax,
correctness conditions and security notions are established, and the
properties of such schemes are analyzed. In Section \ref{sec:agg},
some instantiations are given based on inner-product PE
constructions. Our main construction, XOR-homomorphic IBE, is
presented in Section \ref{sec:cocks}. Non-universal homomorphic
encryption and the abstraction of universal anonymizers is presented
in Section \ref{sec:anon} towards realizing anonymity for our
construction in a weaker setting. Conclusions and future work are presented in
Section \ref{sec:conclusions}.

\section{Preliminaries}\label{sec:prelim}
A quantity is said to be negligible with respect to some parameter
$\lambda$, written $\fname{negl}(\lambda)$, if it is asymptotically
bounded from above by the reciprocal of all polynomials in $\lambda$.

For a probability distribution $D$, we denote by $x
\samplefrom D$ that $x$ is sampled according to $D$. If $S$ is a set,
$y \samplefrom S$ denotes that $y$ is sampled from $x$ according to
the uniform distribution on $S$.

The support of a predicate $f : A \to \{0, 1\}$ for some domain $A$ is
denoted by $\fname{supp}(f)$, and is defined by the set $\{a \in A :
f(a) = 1\}$.

\begin{definition}[Homomorphic Encryption]
A homomorphic encryption scheme with message space $M$ supporting a
class of $\ell$-input circuits $\mathbb{C} \subseteq M^{\ell}
\to M$ is a tuple of PPT algorithms
$(\aname{Gen}, \aname{Enc}, \aname{Dec}, \aname{Eval})$
satisfying the property:

$\forall (\vname{pk}, \vname{sk}) \gets \aname{Gen}(1^\lambda),\;
\forall C \in \mathbb{C},
\forall m_1, \hdots, m_{\ell} \in M$

$\forall c_1, \hdots, c_{\ell} \gets
\aname{Enc}(\vname{pk}, m_1), \hdots, \aname{Enc}(\vname{pk},
m_{\ell})$
\begin{equation*}
 C(m_1, \hdots, m_{\ell}) =
\aname{Dec}(\vname{sk}, \aname{Eval}(\vname{pk}, C, c_i, \hdots,
c_{\ell}))
\end{equation*}
\end{definition}

The following definition is based on \cite{GoldwasserHELecture11},
\begin{definition}[Strongly Homomorphic]
Let $\mathcal{E}$ be a homomorphic encryption scheme with message
space $M$ and class of supported circuits $\mathbb{C} \subseteq \{M^{\ell}
\to M\}$. $\mathcal{E}$ is said to be \emph{strongly homomorphic} iff $\; \forall C \in \mathbb{C},
\;  \forall (\vname{pk}, \vname{sk}) \gets \aname{Gen},\;\forall m_1,
\hdots, m_{\ell} ,  \; \forall c_1, \hdots, c_{\ell} \gets
\aname{Enc}(\vname{pk}, m_1), \hdots, \aname{Enc}(\vname{pk}, m_{\ell})$, the following distributions are statistically indistinguishable
\begin{eqnarray*}
\aname{Enc}(\vname{pk}, C(m_1, \hdots, m_{\ell})) & \ind &
(\aname{Eval}(\vname{pk}, C, c_1, \hdots, c_{\ell}).
\end{eqnarray*}
\end{definition}

\begin{definition}[Predicate Encryption (Adapted from \cite{Katz:2008} Definition 1)]\label{def:pe}
A predicate encryption (PE) scheme for the class of predicates
$\mathcal{F}$ over the set of attributes $A$ and with message space $M$ consists of four
algorithms $\aname{Setup}$, $\aname{GenKey}$, $\aname{Encrypt}$,
$\aname{Decrypt}$ such that:
\begin{itemize}
\item
$\mathbf{\aname{PE.Setup}}$ takes as input the security parameter $1^\lambda$ and
outputs public parameters $\fname{PP}$ and master secret key
$\fname{MSK}$.
\item
$\mathbf{\aname{PE.GenKey}}$ takes as input the master secret key $\fname{MSK}$ and
a description of a predicate $f \in \mathcal{F}$. It outputs a key
$\fname{SK}_f$.
\item
$\mathbf{\aname{PE.Encrypt}}$ takes as input the public parameters $\fname{PP}$,
a message $m \in M$ and an attribute $a \in A$. It returns a
ciphertext $c$. We write this as $c \leftarrow
\aname{Encrypt}(\vname{PP}, a, m)$.
\item
$\mathbf{\aname{PE.Decrypt}}$ takes as input a secret key $\fname{SK}_f$ for a
predicate $f$ and a
ciphertext $c$. It outputs $m$ iff $f(a) = 1$. Otherwise it outputs a
distinguished symbol $\bot$ with all but negligible probability.
\end{itemize}

\end{definition}

\begin{remark}
Predicate Encryption (PE) is known by various terms in the
literature. PE stems from Attribute-Based Encryption (ABE) with Key
Policy, or simply KP-ABE, and differs from it in its support for attribute
privacy. As a result, ``ordinary'' KP-ABE is sometimes known as PE
with public index. Another variant of ABE is CP-ABE (ciphertext
policy) where the encryptor embeds her access policy in the ciphertext
and a recipient must possess sufficient attributes in order to
decrypt. This is the reverse of KP-ABE. In this paper, the
emphasis is placed on PE with
its more standard interpretation, namely KP-ABE with attribute
privacy.
\end{remark}

\section{Homomorphic Predicate Encryption}\label{sec:hpe}
\subsection{Syntax}
Let $M$ be as message space and let $A$ be a set of
attributes. Consider a set of operations $\Gamma_M \subseteq \{M^2
\to M\}$ on the message space, and a set of operations $\Gamma_A
\subseteq \{A^2 \to A\}$ on the attribute space. We denote by $\gamma = \gamma_A \times \gamma_M$ for some $\gamma_A \in \Gamma_A$ and $\gamma_M \in \Gamma_M$ the operation $(A \times M)^2 \to (A \times M)$ given by $\gamma((a_1, m_1), (a_2, m_2)) = (\gamma_A(a_1, a_2), \gamma_M(m_1, m_2))$.  Accordingly, we define the
set of permissible ``gates'' $\Gamma \subseteq \{\gamma_A \times \gamma_M :
\gamma_A \in \Gamma_A, \gamma_M \in \Gamma_M\} \subseteq \{(A \times M)^2 \to (A \times M)\}$\footnote{It is assumed that $\Gamma_A$ and
  $\Gamma_M$ are minimal insofar as $\forall \gamma_A \in \Gamma_A
  \exists \gamma_M \in \Gamma_M \text{ s.t. } \gamma_A \times
  \gamma_M \in \Gamma$ and the converse also holds. In particular, we
  later assume this of $\Gamma_A$.}. Thus, each operation on the plaintexts is
associated with a single (potentially distinct) operation on the
attributes. Finally, we can specify a class of permissible circuits
$\mathbb{C}$ built from
$\Gamma$.

\begin{definition}\label{def:hpe}
A homomorphic predicate encryption (HPE) scheme for the non-empty class of predicates
$\mathcal{F}$, message space $M$, attribute space $A$, and class
of $\ell$-input  circuits $\mathbb{C}$ consists of a tuple
of five PPT algorithms $\aname{Setup}$, $\aname{GenKey}$, $\aname{Encrypt},
\aname{Decrypt}$ and $\aname{Eval}$. such that:
\begin{itemize}
\item
$\mathbf{\aname{HPE.Setup}}$, $\mathbf{\aname{HPE.GenKey}}$,
$\mathbf{\aname{HPE.Encrypt}}$ and $\mathbf{\aname{HPE.Decrypt}}$ are
as specified in Definition \ref{def:pe}.
\item
$\mathbf{\aname{HPE.Eval}}(\fname{PP}, C, c_1, \hdots, c_\ell)$ takes as input the public parameters
$\fname{PP}$, an $\ell$-input circuit $C \in \mathbb{C}$, and
ciphertexts  $c_1 \gets \aname{HPE.Encrypt}(\vname{PP}, a_1, m_1), \hdots,
c_\ell \gets \aname{HPE.Encrypt}(\vname{PP}, a_\ell,
m_\ell)$.

It outputs a ciphertext that encrypts the
attribute-message pair $C((a_1, m_1), \hdots, (a_\ell, m_\ell))$.

\end{itemize}

\end{definition}

Accordingly, the correctness criteria are defined as follows:

\textbf{Correctness conditions:}

For any $(\vname{PP}, \vname{MSK}) \leftarrow
\aname{HPE.Setup(1^\lambda})$, $f \in
\mathcal{F}$, $\vname{SK}_{f} \leftarrow \aname{HPE.GenKey}(\vname{PP},
\vname{MSK}, f)$, $C \in \mathbb{C}$: 
\begin{enumerate}
\item
For any $a \in A, m \in M, c \leftarrow
\aname{HPE.Encrypt}(\fname{PP}, m, a)$:

\begin{equation*}
\aname{HPE.Decrypt}(\fname{SK}_{f}, c) = m \iff f(a) = 1
\end{equation*}
\item
$\forall m_1, \hdots, m_\ell \in M, \; \forall a_1, \hdots, a_\ell
\in A, \;$ $\forall c_1, \hdots, c_\ell \leftarrow \aname{HPE.Encrypt}(\fname{PP},
a_1, m_1), \hdots, \aname{HPE.Encrypt}(\fname{PP}, a_\ell, m_\ell):$\\

$\forall c^\prime \leftarrow \aname{HPE.Eval}(\fname{PP}, C, c_1,
\hdots, c_\ell)$
\begin{enumerate}
\item
\begin{equation*}
\aname{HPE.Decrypt}(\fname{SK}_{f}, c^\prime) = m^\prime \iff f(a^\prime) = 1
\end{equation*}
where $(m^\prime, a^\prime) = C((a_1, m_1), \hdots, (a_\ell, m_\ell))$
\item
\begin{equation*}
|c^\prime| < L(\lambda)
\end{equation*} 

where $L(\lambda)$ is a fixed polynomial derivable from $\vname{PP}$.
\end{enumerate}
\end{enumerate}

The special case of ``predicate only'' encryption \cite{Katz:2008} that excludes
plaintexts (``payloads'') is modelled by setting $M \defeq \{\vec{0}\}$ for a distinguished symbol $\vec{0}$,
and setting $\Gamma \defeq \{\gamma_A \times \mathbf{id_M} : \gamma_A \in
\Gamma_A\}$ where $\mathbf{id_M}$ is the identity operation on $M$.

\subsection{Security Notions}
The security notions we consider carry over from the standard notions
for PE. The basic
requirement is IND-CPA security, which is referred to as
``payload-hiding''. A stronger notion is ``attribute-hiding'' that
additionally entails indistinguishability of attributes. The
definitions are game-based with non-adaptive and adaptive
variants. The former prescribes that the adversary choose its
target attributes at the beginning of the game before seeing the
public parameters, whereas the latter allows the adversary's choice to
be informed by the public parameters and secret key queries.
\begin{definition}\label{def:secgame}
A (H)PE scheme $\mathcal{E}$ is said to be
(fully) attribute-hiding (based on Definition 2 in \cite{Katz:2008}) if an
adversary $\mathcal{A}$ has negligible advantage in the following game:

\begin{enumerate}
\item
In the \textbf{non-adaptive} variant, $\mathcal{A}$ outputs two attributes $a_0$
and $a_1$ at the beginning of the game.
\item
The challenger $\mathcal{C}$ runs Setup($1^\lambda$) and outputs $(\fname{PP},
\fname{MSK})$
\item
\textbf{Phase 1}

$\mathcal{A}$ makes adaptive queries for the secret keys for
predicates $f_1, \hdots, f_k \in \mathcal{F}$ subject to the
constraint that $f_i(a_0) = f_i(a_1)$ for $1 \leq i \leq k$.
\item
\begin{remark} In the stronger \emph{adaptive} variant, $\mathcal{A}$ only
  chooses attributes $a_0$ and $a_1$ at this stage. \end{remark}
\item
$\mathcal{A}$ outputs two messages $m_0$ and $m_1$ of equal length. It
must hold that $m_0 = m_1$ if there is an $i$ such that $f_i(a_0) =
f_i(a_1) = 1$.
\item
$\mathcal{C}$ chooses a random bit $b$, and outputs $c \leftarrow
\aname{Encrypt}(\fname{PP}, a_b, m_b)$
\item
\textbf{Phase 2}

A second phase is run where $\mathcal{A}$ requests secret keys for
other predicates subject to the same constraint as above.
\item
Finally, $\mathcal{A}$ outputs a guess $b^\prime$ and is said to win
if $b^\prime = b$.
\end{enumerate}
\end{definition}

A weaker property referred to as \emph{weakly} attribute-hiding \cite{Katz:2008} requires
that the adversary only request keys for predicates $f$ obeying
$f(a_0) = f(a_1) = 0$.

We propose another model of security for non-universal homomorphic
encryption in Section \ref{sec:anon}.

\subsection{Attribute Operations}
We now characterize HPE schemes based on the properties of their
attribute operations (elements of $\Gamma_A$).
\begin{definition}[Properties of attribute operations]
$\forall f \in \mathcal{F}, \quad \forall a_1, a_2 \in A, \quad
\forall \gamma_A \in \Gamma_A$:
\begin{enumerate}
\item
\begin{equation}\label{eq:minimal}
f(\gamma_A(a_1, a_2)) \implies f(a_1) \land f(a_2)
\end{equation}
\emph{(Necessary condition for IND-CPA security)}

\item
\emph{}
\begin{equation}\label{eq:idempotence}
f(\gamma_A(a_1, a_1)) = f(a_1)
\end{equation}

\item
$\forall d \in A$:

\begin{eqnarray*}
f(a_1) = f(a_2) & \implies & f(\gamma_A(d, a_1)) = f(\gamma_A(d, a_2))\\
& \land &  f(\gamma_A(a_1, d)) = f(\gamma_A(a_2, d))
\end{eqnarray*}
\begin{equation}\label{eq:nonmonotone}
\end{equation}
\emph{(Non-monotone Indistinguishability)}

\item
\begin{equation}\label{eq:monotone}
f(\gamma_A(a_1, a_2)) = f(a_1) \land f(a_2)
\end{equation}
\emph{(Monotone Access)}

\end{enumerate}
\end{definition}

Property \ref{eq:minimal} is a minimal precondition for
payload-hiding i.e. IND-CPA security under both adaptive and
non-adaptive security definitions. 

Property \ref{eq:idempotence} preserves
access under a homomorphic operation on ciphertexts with the same attribute.

Property \ref{eq:nonmonotone} is a necessary condition for full
attribute-hiding.

Property \ref{eq:monotone} enables monotone access; a user only learns a function of a
plaintext if and only if that user has permission to learn the value of that
plaintext. This implies that $(A, \gamma_A)$ cannot be a group unless
$\mathcal{F}$ is a class of constant predicates. In general,
\ref{eq:monotone} implies that $\mathcal{F}$ is monotonic. Monotone access is equivalent
to the preceding three properties collectively; that is
\begin{equation*}
\ref{eq:minimal} \land \ref{eq:idempotence} \land \ref{eq:nonmonotone}
\iff \ref{eq:monotone}
\end{equation*}

\subsubsection{Non-Monotone Access}\label{sec:nonmonotone}
Non-monotone access is trickier to define and to suitably
accommodate in a security definition. It can arise from policies that
involve negation. As an example, suppose that it is permissible for a
party to decrypt data sets designated as either ``geology'' or
``aviation'', but is not authorized to decrypt results with both
designations that arise from homomorphic computations on both data
sets. Of course it is then necessary to strengthen the restrictions on the adversary's
choice of $a_0$ and $a_1$ in the security game. Let $a_0$ and
$a_1$ be the attributes chosen by the adversary. Intuitively, the goal
is to show that any sequence of transitions that leads $a_0$ to a
an element outside the support of $f$, also leads $a_1$ to
an element
outside the support of $f$,
and vice versa. Instead of explicitly imposing this non-triviality
constraint on the adversary's choice of attributes, one may seek to
show that there is no pair of attributes distinguishable under any
$\gamma_A$ and $f \in \mathcal{F}$.  This is captured by the property
of non-monotone
indistinguishability (\ref{eq:nonmonotone}). Trivially, the constant
operations satisfy \ref{eq:nonmonotone}. Of more interest is an
operation that limits homomorphic operations to ciphertexts with the
same attribute. This captures our usual requirements for the
(anonymous) IBE functionality, but it is also satisfactory for many
applications of
general PE where computation need only be performed on ciphertexts with
matching
attributes. To accomplish this, the attribute
space is augmented with a (logical) absorbing element $z$
such that $f(z) = 0\; \forall f \in
\mathcal{F}$. The
attribute operation is defined as follows:
\begin{equation}
\delta(a_1, a_2) = \begin{cases} 
a_1 & \text{ if } a_1 = a_2\\
z & \text{ if } a_1 \neq a_2
\end{cases}
\end{equation}
$\delta$ models the inability to perform homomorphic
evaluations on ciphertexts associated with unequal
attributes (identities in the case of IBE). A scheme with this
operation can only be fully attribute-hiding in a vacuous sense
(it may be such that no restrictions are placed upon the adversary's
choice of $f$ but it is unable to find attributes $a_0$ and $a_1$
satisfying $f(a_0) = f(a_1) = 1$ for any $f$.) This is the case for
anonymous IBE where the predicates are equality relations, and for the
constant map $(a_1, a_2) \mapsto z$ that models the absence of
a homomorphic property, although this is preferably modeled by
appropriately constraining the class of permissible circuits. More generally, such schemes can only be
weakly attribute-hiding because their operations $\gamma_A$ only satisfy a relaxation of
\ref{eq:nonmonotone} given as follows:

\emph{Necessary condition for weakly attribute-hiding}
$\forall a_1, a_2, d \in A$:
\begin{eqnarray*}
f(a_1) = f(a_2) = 0 & \implies & f(\gamma_A(d, a_1)) = f(\gamma_A(d, a_2))\\
& \land &  f(\gamma_A(a_1, d)) = f(\gamma_A(a_2, d))
\end{eqnarray*}
\begin{equation}\label{eq:nonmonotone_weak}
\end{equation}

\begin{remark}
In the case of general schemes not satisfying \ref{eq:nonmonotone}, placing constraints on the adversary's choice of attributes weakens
the security definition. Furthermore, it must be possible for the challenger to efficiently check whether a
pair of attributes satisfies such a condition. Given the added
complications, it is tempting to move to a simulation-based definition
of security. However, this is precluded by the recent impossibility
results of \cite{AGVW12ePrint} in the case of both weakly and fully
attribute-hiding in the NA/AD-SIM models of
security. However, for predicate encryption with public index (the
attribute is not hidden), this has not been ruled out for 1-AD-SIM and
many-NA-SIM where ``1'' and ``many'' refer to the number of
ciphertexts seen by the adversary. See
\cite{AGVW12ePrint,BO12ePrint} for more details. In the context of non-monotone access,
it thus seems more reasonable to focus on predicate encryption with
public index. Our main focus in this work is on schemes that
facilitate attribute privacy, and therefore we restrict our attention
to schemes that at least satisfy \ref{eq:nonmonotone_weak}.
\end{remark}

\subsubsection{Delegate Predicate Encryption}
A primitive presented in \cite{Wei:2009} called ``Delegate Predicate
Encryption'' (DPE) \footnote{Not to be confused with the different notion of
  Delegatable Predicate Encryption.} enables a user to generate an encryption key associated
with a chosen attribute $a \in A$, which does not reveal anything
about $a$. The user can distribute this to certain parties who can then
encrypt messages with attribute $a$ obliviously. The realization in \cite{Wei:2009} is similar to the widely-used
technique of publishing encryptions of ``zero'' in a homomorphic
cryptosystem, which can then be treated as a key. In fact, this technique
is adopted in \cite{Rothblum:2011} to transform a strongly homomorphic
private-key scheme into a public-key one. Generalizing from the
results of \cite{Wei:2009}, this corollary follows from the property
of attribute-hiding

\begin{corollary}
An attribute-hiding HPE scheme is a DPE as defined in \cite{Wei:2009} if there exists a
$\gamma \in \Gamma$ such that $(A \times M, \gamma)$ is unital.
\end{corollary}

\section{Constructions with Attribute Aggregation}\label{sec:agg}
In this section, we give some meaningful examples of attribute
homomorphisms (all which satisfy monotone access)
for some known primitives. We begin with a
special case of PE introduced by Boneh and Waters \cite{Boneh:2007}, which they
call \emph{Hidden Vector Encryption}. In this primitive, a ciphertext embeds a vector $\vec{w}
\in \{0, 1\}^n$ where $n$ is fixed in the public parameters. On the other hand, a secret key corresponds to a
vector $\vec{v} \in V \defeq \{\ast, 0, 1\}^n$ where $\ast$ is
interpreted as a ``wildcard'' symbol or a ``don't care'' (it matches
any symbol). A decryptor who has a secret key for some $\vec{v}$ can
check whether it matches the attribute in a ciphertext.

To formulate in terms of PE, let $A = \{0, 1\}^n$ and define
\begin{equation*}
\mathcal{F} \subseteq \{(w_1, \hdots, w_n) \mapsto
\bigwedge_{i = 1}^{n} (v_i = w_i \lor v_i = \ast) :
\vec{v} \in V\}
\end{equation*}

Unfortunately, we cannot achieve a non-trivial homomorphic variant of HVE that satisfies \ref{eq:monotone}. To see this, consider the HVE class of predicates $\mathcal{F}$ and an operation $\gamma_A$ satisfying \ref{eq:monotone}. For any
$\vec{x}, \vec{y} \in A$, let $\vec{z} = \gamma_A(\vec{x},
\vec{y})$. Now for \ref{eq:monotone} to hold, we must have that
$f(\vec{z}) = f(\vec{x}) \land f(\vec{y})$ for all $f \in
\mathcal{F}$. Suppose $\vec{x}_i \neq \vec{y}_i$ and $\vec{z}_i = \vec{x}_i$. Then there exists an $f \in \mathcal{F}$ with $f(\vec{z}) = f(\vec{x})$ and $f(\vec{z}) \neq f(\vec{y})$. It is necessary to restrict $V$. Accordingly, let $V =
\{\ast, 1\}^n$
Setting the non-equal
elements to 0  yields associativity and
commutativity. Such an operation is equivalent to component-wise logical
AND on the attribute vectors, and we will denote it by
$\land^n$. $(A, \land^n)$ is a semilattice.

Recall that a predicate-only scheme does
not incorporate a payload into ciphertexts. Even such a scheme
$\mathcal{E}$ with
the $\land^n$ attribute homomorphism might find some purpose in real-world
scenarios. One particular application of
$\mathcal{E}$ is
secure data aggregation in Wireless Sensor Networks (WSNs), an area which has
been the target of considerable research (a good survey is \cite{Alzaid:2008}). It
is conceivable that some aggregator nodes may be authorized by the sink
(base station to which packets are forwarded) to read packets matching
certain criteria. An origin sensor node produces an outgoing
ciphertext as follows: (1). It encrypts the attributes describing its
data using $\mathcal{E}$. (2) It encrypts its sensor reading with the
public key of the sink using a \emph{separate} additively (say)
homomorphic public-key cryptosystem. (3) Both ciphertexts are forwarded
to the next hop.

Since an aggregator node receives packets from multiple
sources, it needs to have some knowledge about how to aggregate
them. To this end, the sink can authorize it to apply a particular predicate to incoming
ciphertexts to check for matching candidates for aggregation. One sample policy may be [``REGION1'' $\land$
``TEMPERATURE''']. It can then aggregate ciphertexts matching this policy. Additional
aggregation can be performed by a node further along the route that
has been perhaps
issued a secret key for a predicate corresponding to the more
permissive policy of [``TEMPERATURE'']. In the scenario above, it
would be more ideal if $\mathcal{E}$ were also additively homomorphic
since besides obviating the need to use another
PKE cryptosystem, more control is afforded to aggregators; they receive
the ability to decrypt partial sums, and therefore, to perform (more
involved) statistical computations on the data. 

It is possible to achieve the former case
  from some recent inner-product PE schemes that admit
homomorphisms on both attributes and payload. We focus on two
prominent constructions with different mathematical
structures. Firstly, a construction is examined
by Katz, Sahai and Waters (KSW) \cite{Katz:2008}, which relies on non-standard assumptions on
bilinear groups, assumptions that are justified by the authors in the
generic group model.  Secondly, we focus on a construction presented by
Agrawal, Freeman and Vaikuntanathan (AFV) \cite{Agrawal11} whose security is
based on the learning with errors (LWE) problem. 

In both schemes, an attribute is an element of $\ZZ_m^n$\footnote{In
  \cite{Katz:2008}, $m$ is a product of three large primes and $n$ is
  the security parameter. In \cite{Agrawal11}, $n$ is independent of
  the security parameter and $m$ may be polynomial or superpolynomial in the security
  parameter; in the latter case $m$ is the product
  of many ``small'' primes. We require that $m$ be superpolynomial here.} and a predicate also corresponds to an
element of $\ZZ_m^n$. For $\vec{v} \in \ZZ_m^n$, a predicate
$f_{\vec{v}} : \ZZ_m^n \to \{0, 1\}$ is defined by
\begin{equation*}
f_{\vec{v}}(\vec{w}) = \begin{cases} 1 & \text{ iff } \langle \vec{v},
  \vec{w} \rangle \\ 0 & \text{ otherwise }
\end{cases}
\end{equation*}

Roughly speaking, in a ciphertext, all sub-attributes (in $\ZZ_m$) are blinded by the
same uniformly random ``blinding'' element $b$ \footnote{a scalar in KSW and a
  matrix in AFV}. The decryption algorithm multiplies each
component by the corresponding component in the predicate vector, and
the blinding element $b$ is eliminated when the inner product
evaluates to zero
with all but negligible probability, which allows decryption to
proceed. 

Let $\vec{c_1}$ and $\vec{c_2}$ be ciphertexts that encrypt attributes $a_1$ and
$a_2$ respectively. It can be easily shown that the sum
$\vec{c^\prime} = \vec{c_1} \boxplus$\footnote{$\boxplus$ denotes a
  pairwise sum of the ciphertext components in both schemes}
$\vec{c_2}$ encrypts both $a_1$ and $a_2$ in a somewhat ``isolated''
way. The lossiness is ``hidden'' by the negligible probability of two
non-zero inner-products summing to 0.  For linear aggregation, this can be
repeated a polynomial number of times (or effectively unbounded in
practice) while ensuring correctness with overwhelming probability. While
linear aggregation is sufficient for the WSN scenario, it is
interesting to explore other circuit forms. For the KSW scheme, we
observe that all circuits of polynomial depth can be evaluated with
overwhelming probability. For
AFV, the picture is somewhat similar to the fully homomorphic
schemes based on LWE such as \cite{BV11a,Brakerski2011b} but without
requiring multiplicative gates. 

While there are motivating scenarios for aggregation on the
attributes, in many cases it is adequate or preferable to restrict
evaluation to ciphertexts with matching attributes; that is, by means of the $\delta$
operation defined in Section \ref{sec:nonmonotone}. Among these cases
is anonymous IBE. In the next section, we introduce an IBE
construction that supports an unbounded XOR homomorphism, prove
that it is strongly homomorphic and then investigate anonymous
variants.

\section{Main Construction: XOR-Homomorphic IBE}\label{sec:cocks}
In this section, an XOR-homomorphic IBE scheme is presented whose security is based
on the quadratic residuosity assumption. Therefore, it is similar in
many respects to the Goldwasser-Micali (GM) cryptosystem \cite{GM82}, which
is well-known to be XOR-homomorphic. Indeed, the GM scheme has found
many practical applications due to its homomorphic property. In
Section \ref{sec:app}, we show how many of these applications benefit
from an XOR-homomorphic scheme in the identity-based setting.

Our construction derives from the IBE scheme due to Cocks
\cite{COCKS01} which has a security reduction to the quadratic
residuosity problem. To the best of our knowledge, a homomorphic
variant has not been
explored to date.

\subsection{Background}
Let $m$ be an integer. A quadratic residue in the residue ring $\ZZ_m$ is an integer $x$ such
that $x \equiv y^2 \mod{m}$ for some $y \in \ZZ_m$. The set of
quadratic residues in $\ZZ_m$ is denoted $\QR(m)$. If $m$ is prime, it
easy to determine whether any $x \in \ZZ_m$ is a quadratic residue.

Let $N = pq$ be a composite modulus where $p$ and $q$ are prime. Let
$x \in \ZZ$. We
write $\jacobi{x}{N}$ to denote the Jacobi symbol of $x \mod{N}$. The
subset of integers with Jacobi symbol +1 (resp. -1) is denoted
$\ZZ_N[+1]$ (resp. $\ZZ_N[-1])$. The quadratic
residuosity problem is to determine, given input $(N, x \in \ZZ_N[+1])$,
whether $x \in \QR(N)$, and it is believed to be intractable.

Define the encoding $\nu : \{0, 1\} \to \{-1, 1\}$ with $\nu(0) = 1$
and $\nu(1) = -1$. Formally, $\nu$ is a group isomorphism between
$(\ZZ_2, +)$ and $(\{-1, 1\}, \ast)$.

In this section, we build on the results of \cite{AG09} and
therefore attempt to maintain consistency with their notation where
possible. As in \cite{AG09}, we let $H : \{0, 1\}^\ast \to \ZZ^\ast_N[+1]$
be a full-domain hash. A message bit is mapped to an element of $\{-1,
1\}$ via $\nu$ as defined earlier (0 (1 resp.) is encoded as 1 (-1 resp.)).

\subsection{Original Cocks IBE Scheme}
\begin{itemize}
\item
$\mathbf{\aname{CocksIBE}.Setup}(1^\lambda)$: 
\begin{enumerate}
\item
Repeat:
$p, q \samplefrom \aname{RandPrime}(1^\lambda)$
Until: $p \equiv q \equiv 3 \pmod{4}$
\item
$N \gets pq$
\item
Output $(\vname{PP} := N, \vname{MSK} := (p, q))$
\end{enumerate}
\item
$\mathbf{\aname{CocksIBE}.KeyGen}(\vname{PP}, \vname{MSK}, \vname{id})$: 
\begin{enumerate}
\item
Parse $\vname{MSK}$ as $(p, q)$.
\item
$a \gets H(\vname{id})$
\item
$r \gets a^{\frac{N + 5 - p - q}{8}} \pmod{N}$\\
($\therefore r^2 \equiv a \pmod{N}$ or $r^2 \equiv -a \pmod{N}$)
\item
Output $\vname{sk}_{\vname{id}} := (\vname{id}, r)$
\end{enumerate}
\item
$\mathbf{\aname{CocksIBE}.Encrypt}(\vname{PP}, \vname{id}, b)$: 
\begin{enumerate}
\item
$a \gets H(\vname{id})$
\item
$t_1, t_2 \samplefrom \ZZ^\ast_N[\nu(b)]$
\item
Output $\vec{\psi} := (t_1 + at_1^{-1}, t_2 - at_2^{-1})$
\end{enumerate}
\item
$\mathbf{\aname{CocksIBE}.Decrypt}(\vname{PP}, \vname{sk}_{\vname{id}}, \vec{\psi})$: 
\begin{enumerate}
\item
Parse $\vec{\psi}$ as $(\vec{\psi}_1, \vec{\psi}_2)$
\item
Parse $\vname{sk}_{\vname{id}}$ as $(\vname{id}, r)$
\item
$a \gets H(\vname{id})$
\item
If $r^2 \equiv a \pmod{N}$, set $d \gets \vec{\psi}_1$. Else if $r^2
\equiv -a \pmod{N}$, set $d \gets \vec{\psi}_2$. Else output $\bot$ and abort.
\item
Output $\nu^{-1}(\jacobi{d + 2r}{N})$
\end{enumerate}
\end{itemize}

The above scheme can be shown to be adaptively secure in the random oracle
model assuming the hardness of the quadratic residuosity problem.

\subsubsection{Anonymity}\label{sec:cocks_orig_anon}
Cocks' scheme is not anonymous. Boneh et al. \cite{Boneh:2004} report
a test due to Galbraith that enables an attacker to distinguish the
identity of a ciphertext. This is achieved with overwhelming
probability given multiple ciphertexts. It is shown by Ateniese and
Gasti \cite{AG09} that there is no ``better'' test for attacking
anonymity. Briefly, let $a = H(\vname{id})$ be the public key derived from the identity
$\vname{ID_a}$. Let $c$ be a
ciphertext in the Cocks' scheme. Galbraith's test is defined as
\begin{equation*}
\fname{GT}(a, c, N) = \jacobi{c^2 - 4a}{N}
\end{equation*}
Now if $c$ is a ciphertext encrypted with $a$, then $\fname{GT}(a, c,
N) = +1$ with all but negligible probability. For $b \in
\ZZ_N^\ast$ such that $b \neq a$, the value $\fname{GT}(b, c, N)$ is
statistically close to the
uniform distribution on $\{-1, 1\}$. Therefore, given multiple
ciphertexts, it can be determined with overwhelming probability
whether they correspond to a particular identity. 

\subsection{XOR-homomorphic Construction}
Recall that a ciphertext in the Cocks scheme consists of two elements in $\ZZ_N$. Thus, we have \[(c, d) \gets \aname{CocksIBE.Encrypt}(\vname{PP}, \vname{id}, b) \in \ZZ_N^2\] for some identity $\vname{id}$ and bit $b \in \{0, 1\}$. Also recall that only one element is actually used for decryption depending on whether $a := H(\vname{id}) \in \QR(N)$ or $-a \in \QR(N)$. If the former holds, it follows that a decryptor has a secret key $r$ satisfying $r^2 \equiv a \pmod{N}$. Otherwise, a secret key $r$ satisfies $r^2 \equiv -a \pmod{N}$. To simplify the description of the homomorphic property, we will assume that $a \in \QR(N)$ and therefore omit the second ``component'' $d$ from the ciphertext. In fact, the properties hold analogously for the second ``component'' by simply replacing $a$ with $-a$.

In the homomorphic scheme, each ``component'' of the ciphertext is represented by a pair of elements in $\ZZ_N^2$ instead of a single element as in the original Cocks scheme. As mentioned, we will omit the second such pair for the moment.
Consider
the following encryption algorithm $E_a$ defined by
\vspace{5pt}

$\aname{E_a}(b : \{0, 1\}):$
\begin{algorithmic}
\State \indent$t
\samplefrom \ZZ^\ast_N[\nu(b)]$
\State \indent \Return $(t + at^{-1}, 2) \in \ZZ^2_N$.
\end{algorithmic}
Furthermore, define the decryption function $D_a(\vec{c}) = \nu^{-1}(c_0 + rc_1)$. The homomorphic operation
$\boxplus : \ZZ^2_N \times \ZZ^2_N \to \ZZ^2_N$ is defined as follows:
\begin{equation}
\vec{c} \boxplus \vec{d} = (c_0d_0 + ac_1d_1, c_0d_1 + c_1d_0)
\end{equation}
It is easy to see that $D_a(\vec{c} \boxplus
\vec{d}) = D_a(\vec{c}) \oplus D_a(\vec{d})$:
\begin{eqnarray}\label{eq:xh_derivation}
D_a(\vec{c} \boxplus
\vec{d}) &=& D_a((c_0d_0 + ac_1d_1, c_0d_1 + c_1d_0)) \nonumber \\
&=&  \nu^{-1}((c_0d_0 + ac_1d_1) + r(c_0d_1 + c_1d_0)) \nonumber \\
&=&\nu^{-1}(c_0d_0 + rc_0d_1 + rc_1d_0 +
r^2c_1d_1) \nonumber \\
&=& \nu^{-1}((c_0 + rc_1)(d_0 + rd_1)) \nonumber \\
&=& \nu^{-1}(c_0 + rc_1) \oplus \nu^{-1}(d_0 + rd_1) \nonumber \\
&=& D_a(\vec{c}) \oplus D_a(\vec{d})
\end{eqnarray}

Let $R_a = \ZZ_N[x] / (x^2 - a)$ be a quotient of the polynomial ring
$R = \ZZ_N[x]$.  It is more natural and convenient to
view ciphertexts as elements of $R_a$ and the homomorphic operation as
multiplication in $R_a$. Furthermore, decryption equates to
evaluation at the point $r$. Thus the homomorphic evaluation of two
ciphertext polynomials $c(x)$ and $d(x)$ is simply $e(x) = c(x) \ast d(x)$
where $\ast$ denotes multiplication in $R_a$. Decryption becomes
$\nu^{-1}(e(r))$. Moreover, Galbraith's test is generalized straightforwardly to the ring $R_a$:
\begin{equation*}
\fname{GT}(a, c(x)) = \jacobi{c_0^2 - c_1^2a}{N}.
\end{equation*}
We now formally describe our variant of the Cocks
scheme that supports an XOR homomorphism.

\begin{remark}
We have presented the scheme in accordance with Definition
\ref{def:hpe} for consistency with the rest of the paper. Therefore,
it uses the circuit formulation, which we would typically consider
superfluous for a group homomorphic scheme.
\end{remark}

Let $\mathbb{C} \defeq \{\vec{x}
\mapsto \langle \vec{t}, \vec{x} \rangle : \vec{t} \in \ZZ^{\ell}_2\}  \subset \ZZ^{\ell}_2 \to \ZZ_2$ be the class of
arithmetic circuits characterized by linear functions over $\ZZ_2$ in
$\ell$ variables. As such, we associate a representative vector $V(C)
\in \ZZ^{\ell}_2$ to every circuit $C \in \mathbb{C}$. 
In order to obtain a strongly homomorphic scheme, we use the standard technique of re-randomizing the evaluated
ciphertext by homomorphically adding an encryption of zero.
\begin{itemize}
\item
$\mathbf{\aname{xhIBE}.Encrypt}(\vname{PP}, \vname{id}, b)$: 
\begin{enumerate}
\item
$a \gets H(\vname{id})$
\item
As a subroutine (used later), define

$E(\vname{PP}, a, b)$:
\begin{enumerate}
\item
$t_1, t_2 \samplefrom \ZZ^\ast_N[\nu(b)]$
\item
$g_1, g_2 \samplefrom \ZZ^\ast_N$
\item
$c(x) \gets (t_1 + ag_1^2t_1^{-1}) + 2g_1x \in \ZZ_N[x]$
\item
$d(x) \gets (t_2 + ag_2^2t_2^{-1}) + 2g_2x \in \ZZ_N[x]$
\item
Repeat steps (a) - (d) until $\fname{GT}(a, c(x)) = 1$ and $\fname{GT}(-a,
d(x)) = 1$.
\item
Output $(c(x), d(x))$
\end{enumerate}
\item
Output $\vec{\psi} := (E(\vname{PP}, a, b), a)$
\end{enumerate}
\item
$\mathbf{\aname{xhIBE}.Decrypt}(\vname{PP}, \vname{sk}_{\vname{id}}, \vec{\psi})$: 
\begin{enumerate}
\item
Parse $\vec{\psi}$ as $(c(x), d(x), a)$
\item
Parse $\vname{sk}_{\vname{id}}$ as $(\vname{id}, r)$
\item
If $r^2 \equiv a \pmod{N}$ and $\fname{GT}(a, c(x)) = 1$, set $e(x) \gets c(x)$. Else if $r^2
\equiv -a \pmod{N}$ and $\fname{GT}(-a, c(x)) = 1$, set $e(x) \gets d(x)$. Else output $\bot$
and abort.
\item
Output $\nu^{-1}(\jacobi{e(r)}{N})$
\end{enumerate}
$\mathbf{\aname{xhIBE}.Eval}(\vname{PP}, C, \vec{\psi_1}, \hdots, \vec{\psi}_{\ell})$: 
\begin{enumerate}
\item
Parse $\vec{\psi_i}$ as $(c_i(x), d_i(x), a_i)$ for $1 \leq i \leq
\ell$
\item
If $a_i \neq a_j$ for $1 \leq i, j \leq \ell$, abort with $\bot$.
\item
Let $a = a_1$ and let $R_a = \ZZ_N[x] / (x^2 - a)$
\item
$v \gets V(C)$
\item
$J \gets \{1 \leq i \leq \ell : v_i = 1\}$
\item
$(c^\prime(x), d^\prime(x)) \gets (\prod_{i \in J}
  c_i(x) \mod{(x^2 - a)}, \prod_{i \in I}
  d_i(x)) \mod{(x^2 + a)}$

\item
$(c_z(x), d_z(x)) \gets E(\vname{PP}, a, 0)$ ($E$ is defined as a
subroutine in the specification of $\aname{xhIBE.Encrypt}$)
\item
Output $(c^\prime(x) \ast c_z(x) \mod{(x^2 - a)}, d^\prime(x) \ast
d_z(x) \mod{(x^2 + a)}, a)$.
\end{enumerate}
\end{itemize}

We now prove that our scheme is group homomorphic and strongly
homomorphic. A formalization of group homomorphic public-key schemes is given in
\cite{Armknecht:2010}. Our adapted definition for the PE setting
raises some subtle points. The third requirement in
\cite{Armknecht:2010} is more difficult to formalize for general PE;
we omit it from the definition here and leave a complete formalization to
Appendix \ref{ap:ghe}. We remark that this property which relates to
distinguishing ``illegitimate ciphertexts'' during decryption is not
necessary to achieve IND-ID-CPA security.
\begin{definition}[Adapted from Definition 1 in \cite{Armknecht:2010}]
\label{def:ghom}
Let $\mathcal{E} = (G, K, E, D)$ be a PE scheme with message space
$M$, attribute space $A$, ciphertext space $\hat{\mathcal{C}}$ and
class of predicates $\mathcal{F}$. The
scheme $\mathcal{E}$ is group homomorphic with respect to a non-empty
set of attributes $A^\prime \subseteq A$ if for
every $(\vname{PP}, \vname{MSK}) \gets G(1^\lambda)$, every $f \in
\mathcal{F} : A^\prime \subseteq \fname{supp}(f)$, and every $\vname{sk}_f \gets
K(\vname{MSK}, f)$, the
message space $(M, \cdot)$ is a non-trivial group, and there is a binary operation
$\boxdot : \hat{\mathcal{C}}^2 \to
\hat{\mathcal{C}}$ such that the following properties are satisfied
for the restricted ciphertext space
$\hat{\mathcal{C}_f} = \{c \in
\hat{\mathcal{C}} : D_{\vname{sk}_f}(c) \neq \bot\}$:

\begin{enumerate}
\item
The set of all encryptions $\mathcal{C} := \{c \in
\hat{\mathcal{C}_f} \mid c \gets E(\vname{PP}, a, m), a \in
A^\prime, m \in M\}$ under attributes in $A^\prime$ is a non-trivial
group under the operation $\boxdot$.
\item
The restricted decryption $D_{\vname{sk}_f}^\ast := D_{\vname{sk}_f |
  \mathcal{C}}$ is surjective and $\forall c, c^\prime \in
\mathcal{C} \quad D_{\vname{sk}_f}(c \boxdot c^\prime) =
D_{\vname{sk}_f}(c) \cdot D_{\vname{sk}_f}(c^\prime)$.
\item
\textbf{IBE only (generalized in Appendix \ref{ap:ghe})}
If $\mathcal{E}$ is an IBE scheme, then $\hat{\mathcal{C}_f}$ is also
required to be a 
group, and it is required to be computationally indistinguishable from $\mathcal{C}$; that is:
\begin{equation*}
\{(\vname{PP}, f, \vname{sk}_f, S, c) \mid c
\samplefrom \mathcal{C}, S \subset \{\vname{sk}_g \gets K(g) : g \in \mathcal{F}\}\} \cind \{(\vname{PP}, f, \vname{sk}_f, S, \hat{c}) \mid \hat{c}
\samplefrom \hat{\mathcal{C}_f}, S \subset \{\vname{sk}_g \gets K(g) : g \in \mathcal{F}\}\}.
\end{equation*}
\end{enumerate}
\end{definition}
Informally, the above definition is telling us that for a given subset
of attributes $A^\prime$ satisfying a predicate $f$, the set of
honestly generated encryptions under
these attributes forms a group that is epimorphic to the plaintext
group.  It does not say anything about ciphertexts that are not
honestly generated except in the case of IBE, where we require that all
ciphertexts that do not decrypt to $\bot$ under a secret key are indistinguishable. 

For the remainder of this section, we show that $\aname{xhIBE}$
fulfills the definition of a group homomorphic scheme, and that it is
IND-ID-CPA secure under the quadratic residuosity assumption in the
random oracle model. To simplify the
presentation of the proofs, additional notation is needed. In particular, we inherit the
notation from \cite{AG09}, and generalize it to the ring
$R_a$.

Define the subset $G_a \subset R_a$ as follows:
\begin{equation*}
G_a = \{c(x) \in R_a : \fname{GT}(a, c(x)) = 1\}
\end{equation*}
Define the subset $S_a \subset G_a$\footnote{This definition is stricter than its analog
  in \cite{AG09} in that all elements are in $G_a$. This definition here corrects an error in \cite{CHT2013af} where $h \in \ZZ^\ast_N$ instead of $h \in \ZZ_N$.}:
\begin{equation*}
S_a = \{2hx + (t + ah^2t^{-1}) \in G_a \mid h \in \ZZ_N, t, (t + ah^2t^{-1})
\in \ZZ_N^\ast\}
\end{equation*}
We have the following simple lemma: 

\begin{lemma}\label{lemma:alg}\hspace{1pt}\newline
\begin{enumerate}
\item
$(G_a, \ast)$ is a multiplicative group in $R_a$.
\item
$(S_a, \ast)$ is a subgroup of $G_a$
\end{enumerate}
\end{lemma}
\begin{proof}
We must show that $G_a$ is closed under $\ast$. Let $c(x),
d(x) \in G_a$, and let $e(x) = c(x) \ast d(x)$.
\begin{eqnarray*}\fname{GT}(a, e(x)) &=& \jacobi{e_0^2 - ae_1^2}{N} \\
&=& \jacobi{(c_0d_0 + ac_1d_1)^2 - a(c_0d_1 + c_1d_0)^2}{N} \\
&=& \jacobi{(c_0^2 - ac_1^2)(d_0^2 - ad_1^2)}{N}\\
&=& \jacobi{(c_0^2 - ac_1^2)}{N} \jacobi{(d_0^2 - ad_1^2)}{N}\\
&=& \fname{GT}(a, c(x)) \cdot \fname{GT}(a, d(x))\\
& = & 1
\end{eqnarray*}
Therefore, $e(x) \in G_a$.

It remains to show that every element of $G_a$ is a
unit. Let $z = c_0^2 -ac_1^2 \in \ZZ_N$. An inverse $d_1x + d_0$ of $c(x)$ can be computed by setting
$d_0 =
\frac{c_0}{z}$
and $d_1 = \frac{-c_1}{z}$ if it holds that $z$
is invertible in $\ZZ_N$. Indeed such a $d_1x + d_0$ is in $G_a$. Now if
$z$ is not invertible in $\ZZ_N$ then $p | z$
or $q | z$, which implies that $\jacobi{z}{p} = 0$ or $\jacobi{z}{q} =
0$. But $\fname{GT}(a, c(x)) = \jacobi{z}{N} = \jacobi{z}{p}\jacobi{z}{q} = 1$
since $c(x) \in G_a$. Therefore, $z$ is a unit in $\ZZ_N$, and $c(x)$
is a unit in $G_a$.

Finally, to prove (2), note that the members of $S_a$ are exactly the
elements $c(x)$ such that $c_0^2 - c_1^2a$ is a square, and it is easy
to see that this is preserved under $\ast$ in $R_a$.
\qed
\end{proof}

We will also need the following corollary
\begin{corollary}[Extension of Lemma 2.2 in \cite{AG09}]\label{cor:dist}
The distributions $\{(N, a, t + ah^2t^{-1}, 2h) : N \gets
\aname{Setup}(1^\lambda), a \samplefrom \ZZ^\ast_N[+1], t, h \samplefrom \ZZ^\ast_N)\}$ and $\{(N, a, z_0, z_1) : N \gets
\aname{Setup}(1^\lambda), a \samplefrom \ZZ^\ast_N[+1], z_0 + 
z_1x \samplefrom G_a \setminus S_a\}$ are indistinguishable assuming the hardness of the quadratic residuosity problem.
\end{corollary}

\begin{proof}
The corollary follows immediately from Lemma 2.2 in \cite{AG09}
Let $\mathcal{A}$ be an efficient adversary that distinguishes both
distributions. Lemma 2.2 in \cite{AG09} shows that the distributions $d_0 := (\{(N, a, t + at^{-1}) : N \gets
\aname{Setup}(1^\lambda), a \samplefrom \ZZ^\ast_N[+1], t\}$ and $d_1 := \{(N, a, z_0) : N \gets
\aname{Setup}(1^\lambda), a \samplefrom \ZZ^\ast_N[+1], z_1x + z_0
\samplefrom G_a \setminus S_a \mid z_2 = 2\}$ are
indistinguishable. Given a sample $(N, a, c)$, the simulator generates
$h \samplefrom \ZZ^\ast_N$ and computes $b := h^{-2}a$. It passes the
element $(N, b, c, 2h)$ to $\mathcal{A}$. The simulator aborts with
the output of $\mathcal{A}$. 
\qed
\end{proof}

\begin{theorem}
$\aname{xhIBE}$ is a group homomorphic scheme with respect to the
group operation of $(\ZZ_2, +)$.
\end{theorem}
\begin{proof}
Let $a = H(\vname{id})$ for any valid identity string
$\vname{id}$. Assume that the secret key $r$
satisfies $r^2 \equiv a \mod{N}$.  The analysis holds analogously if
$r^2 \equiv -a \mod{N}$; therefore, we omit the second component of
the ciphertexts for simplicity.

By definition, $S_a = \{c(x) \in R_a \mid \vec{\psi} := (c(x), d(x), a) \gets
\aname{xhIBE.Encrypt}(\vname{PP}, \vname{id}, m), m \in M\}$. By corollary \ref{cor:dist}, it holds that $S_a \cind G_a$ without the master secret key. The decryption algorithm only outputs $\bot$ on input $\vec{\psi} := (c(x), d(x), a)$ if $c(x) \notin G_a$ or $d(x) \notin G_{-a}$. Thus, omitting the second component, we have that $S_a$ corresponds to $\mathcal{C}$ and $G_a$ corresponds to $\hat{\mathcal{C}}_f$ in Definition \ref{def:ghom} (in this case $f$ is defined as $f(\vname{id}^\prime) = 1$ iff $\vname{id^\prime} = \vname{id}$).  It follows that the third requirement of Definition \ref{def:ghom} is satisfied. 

By Lemma \ref{lemma:alg}, $G_a$ is a group and $S_a$ is a non-trivial
subgroup of $G_a$. The surjective homomorphism between $\mathcal{C} :=
S_a$ and
$M := \ZZ_2^\ast$ has already been shown in the correctness derivation
in equation \ref{eq:xh_derivation}. This completes the proof.
\qed
\end{proof}
\begin{remark}
It is straightforward to show that $\aname{xhIBE}$ also meets the criteria
for a shift-type homomorphism as defined in \cite{Armknecht:2010}.
\end{remark}

\begin{corollary}\label{cor:shom}
$\aname{xhIBE}$ is \emph{strongly homomorphic}.
\end{corollary}
\begin{proof}
Any group homomorphic
scheme can be turned into a strongly homomorphic scheme by
rerandomizing an evaluated ciphertext. Indeed
this follows from Lemma 1 in \cite{Armknecht:2010}. Rerandomization is
achieved by multiplying the evaluated ciphertext by an
encryption of the identity, as in $\aname{xhIBE.Eval}$. Details
follow for completeness. 

Let $\vname{id}$ be an identity and let $a = H(\vname{id})$. For any circuit $C \in \mathbb{C}$, any messages $b_1, \hdots, b_\ell$ and ciphertexts $\vec{\psi_1},
\hdots, \vec{\psi_\ell} \gets \aname{xhIBE.Encrypt}(\fname{PP},
b_1, \vname{id}), \hdots, \linebreak[1] \aname{xhIBE.Encrypt}(\fname{PP},
b_\ell, \vname{id})$, we have
\begin{equation*}(c^\prime(x), d^\prime(x), a) \gets
\mathbf{\aname{xhIBE}.Eval}(\vname{PP}, C, \vec{\psi_1}, \hdots,
\vec{\psi_\ell}).\end{equation*}
From the last step of $\aname{xhIBE.Eval}$,
we see that
$c^\prime(x) \gets c^{\prime\prime}(x) \ast r(x)$ where $r(x)
\samplefrom {S_a}^{(0)}$ and $c^{\prime\prime}(x)$ is the result of the
  homomorphic evaluation. Suppose that $c^{\prime\prime}(x)$ encrypts
  a bit $b$. Since $S_a$ is a group, it follows that $c^\prime(x)$ is
  uniformly distributed in the coset $S_a^{(b)}$ (of the subgroup $S_a^{(0)}$) and is
thus distributed according to a ``fresh'' encryption of $b$.
\qed
\end{proof}

\begin{theorem}
$\aname{xhIBE}$ is IND-ID-CPA secure in the random oracle model under
  the quadratic residuosity assumption.
\end{theorem}
\begin{proof}
Let $\mathcal{A}$ be an adversary that breaks the IND-ID-CPA security of $\aname{xhIBE}$. We use
$\mathcal{A}$ to construct an algorithm $\mathcal{S}$ to break the
IND-ID-CPA security of the Cocks scheme with the same advantage. $\mathcal{S}$ proceeds as
follows:
\begin{enumerate}
\item
Uniformly sample an element $h \samplefrom \ZZ_N^\ast$. Receive the public parameters $\vname{PP}$ from the challenger
$\mathcal{C}$ and pass them to $\mathcal{A}$.
\item
$\mathcal{S}$ answers a query to $H$ for identity $\vname{id}$ with $H^\prime(\vname{id}) \cdot h^{-2}$ where
$H^\prime$ is $\mathcal{S}$'s random oracle. The responses are
uniformly distributed in $\ZZ_N[+1]$.
\item
$\mathcal{S}$ answers a key generation query for $\vname{id}$ with the
response $K(\vname{id}) \cdot h^{-1}$ where $K$ is its key generation
oracle.
\item
When $\mathcal{A}$ chooses target identity $\vname{id}^\ast$,
$\mathcal{S}$ relays $\vname{id}^\ast$ to $\mathcal{C}$. Assume w.l.o.g that $H$ has been queried for $\vname{id}$,
and that $\mathcal{A}$ has not made a secret key query for $\vname{id}^\ast$. Further key generation requests are handled subject to the
condition that $\vname{id} \neq \vname{id}^\ast$ for a requested
identity $\vname{id}$.
\item
Let $a = H(\vname{id}^\ast)$. On receiving a challenge ciphertext $(c, d)$ from $\mathcal{C}$, compute
$c(x) \gets 2hx + c \in R$ and $d(x) \gets (2hx + d) \ast r(x) \in R$ where
$r(x) \samplefrom S_{-a}^{(0)}$ and $S_{-a}^{(0)}$ is the second
component of the set of legal
encryptions of 0. From corollary \ref{cor:shom}, $d(x)$ is uniformly
distributed in $S_{-a}^{(b)}$ where the ciphertext $(c, d)$ in the Cocks
scheme encrypts the bit $b$. It follows that $(c(x), d(x))$ is a
perfectly simulated encryption of $b$ under identity $\vname{id}^\ast$
in $\aname{xhIBE}$. Give $(c(x), d(x))$ to $\mathcal{A}$.
\item
Output $\mathcal{A}$'s guess $b^\prime$.
\end{enumerate}

Since the view of $\mathcal{A}$ in an interaction with $\mathcal{S}$
is indistinguishable from its view in the real game, we conclude that
the advantage of $\mathcal{S}$ is equal to the advantage of $\mathcal{A}$.

\qed
\end{proof}

In the next section, attention is drawn to obtaining an anonymous variant of our construction.

\section{Anonymity}\label{sec:anon}
Cocks' scheme is notable as one of the few IBE schemes that
do not rely on pairings. Since it appeared, there have been efforts to
reduce its ciphertext size and make it anonymous. Boneh, Gentry and
Hamburg \cite{BGH07} proposed a scheme with some elegant ideas that
achieves both anonymity and a much reduced ciphertext size for
multi-bit messages at the
expense of performance, which is $O(n^4)$ for encryption and $O(n^3)$
for decryption (where $n$ is the security
parameter). Unfortunately the homomorphic property is lost in this
construction.

As mentioned earlier (cf. Section
\ref{sec:cocks_orig_anon}), another approach due to Ateniese and
Gasti \cite{AG09} achieves anonymity and preserves performance,
but its per-bit ciphertext expansion is much higher than in
\cite{BGH07}. However, an advantage of this scheme is that it is
universally anonymous (anyone can anonymize the message, not merely
the encryptor \cite{Hayashi:2005}).

On the downside, anonymizing
according to this scheme breaks the homomorphic property of our
construction, which depends crucially on the public key $a$. More
precisely, what is forfeited is the \emph{universal homomorphic}
property mentioned in the introduction (i.e. anyone can evaluate
on the ciphertexts without additional information). There are
applications where an evaluator is aware of the attribute(s)
associated with ciphertexts, but anonymity is desirable to prevent any
other parties in the system learning about such attributes.  This motivates a variant of
HPE, which we call non-universal HPE, denoted by
$\vname{HPE}_{\bar{U}}$. 

\subsection{Non-Universal HPE}
\textbf{Motivation}
``Non-universal'' homomorphic encryption is proposed for schemes that support attribute privacy but require some information that is derivable from the public key (or attribute in the case of PE) in order to perform homomorphic evaluation. Therefore, attribute privacy must be surrendered to an evaluator. If this is acceptable for an application, while at the same time there is a requirement to hide the target recipient(s) from other entities in the system, then ``non-universal'' homomorphic encryption may be useful. Consider the following informal scenario. Suppose a collection of parties $P_1, \hdots, P_\ell$ outsource a computation on their encrypted data sets to an untrusted remote server $S$. Suppose $S$ sends the result (encrypted) to an independent database $\vname{DB}$ from which users can retrieve the encrypted records. For privacy reasons, it may be desirable to limit the information that $\vname{DB}$ can learn about the attributes associated with the ciphertexts retrieved by certain users. Therefore, it may desirable for the encryption scheme to provide attribute privacy. However, given the asymmetric relationship between the delegators $P_1, \hdots, P_\ell$ and the target recipient(s), it might be acceptable for $S$ to learn the target attribute(s) provided there is no collusion between $S$ and $\vname{DB}$. In fact, the delegators may belong to a different organization than the recipient(s).

In this paper, we introduce a syntax and security model for non-universal homomorphic IBE. The main change in syntax entails an
additional input $\alpha$ that is supplied to the $\aname{Eval}$
algorithm. The input $\alpha \in \{0, 1\}^d$
(where $d
= poly(\lambda)$) models the additional information needed to
compute the homomorphism(s). A description of an efficient map $Q_A : A
\to \{0, 1\}^d$ is included in the public parameters. We say that two attributes(i.e. identities in IBE) $a_1, a_2 \in A$ satisfying $Q_A(a_1) = Q_A(a_2)$ belong to the same attribute class. 

One reason that the proposed syntax is not general enough for arbitrary PE functionalities is that it only facilitates evaluation on ciphertexts whose attributes are in the same attribute class, which suffices for (relatively) simple functionalities such as IBE.

We now formulate the security notion of attribute-hiding for
non-universal homomorphic IBE. Our security model provides the
adversary with an evaluation oracle whose identity-dependent input
$\alpha$ is fixed when the challenge is produced. Accordingly,  for a
challenge identity $\vname{id} \in A$, and binary string $\alpha = Q_A(\vname{id}) \in \{0,
1\}^d$,  the
adversary can query $\aname{{IBE}_{\bar{U}}.Eval}(\vname{PP}, \alpha, \cdot,
\cdot)$ for any circuit in $\mathbb{C}$ and any $\ell$-length sequence
of ciphertexts. 

Formally, consider the experiment

\textbf{Experiment} $\boldsymbol{\aname{{\bar{U}Priv}}(\mathcal{A}_1,
  \mathcal{A}_2)}$\footnote{In the random oracle model, the adversary
  is additionally given access to a random oracle. This is what the
  results in this paper will use.}
\begin{algorithmic}
\State
\indent$(\vname{PP}, \vname{MSK}) \gets
\aname{IBE.Setup(1^\lambda})$
\State
\indent$(\vname{id}_0, m_0), (\vname{id}_1, m_1), \sigma \gets
\mathcal{A}_1^{\aname{{IBE}_{\bar{U}}.KeyGen}(\vname{MSK},
  \cdot)}(\vname{PP})$ \Comment $\sigma$ denotes the adversary's state
\State
\indent $b \samplefrom \{0, 1\}$
\State
\indent $\alpha \gets Q_A(\vname{id}_b)$
\State
\indent $c \gets \aname{IBE.Encrypt}(\vname{PP}, \vname{id}_b, m_b)$
\State
\indent $b^\prime \gets \mathcal{A}_2^{\aname{{IBE}_{\bar{U}}.KeyGen}^\ast(\vname{MSK},
  \cdot), \aname{{IBE}_{\bar{U}}.Eval}(\vname{PP}, \alpha, \cdot, \cdot)}(\vname{PP}, c, \sigma)$
\State
\indent \Return 1 iff $b^\prime = b$ and 0 otherwise.
\end{algorithmic}

Define the advantage of an adversary $\mathcal{A} := (\mathcal{A}_1, \mathcal{A}_2)$
in the above experiment for a $\vname{IBE}_{\bar{U}}$ scheme $\mathcal{E}$ as follows:
\begin{equation*}
\adv{\mathcal{E}}{\aname{{\bar{U}Priv}}}(\mathcal{A}) =
\pr\big[\aname{\bar{U}Priv}(\mathcal{A}) \implies 1\big] - \frac{1}{2}.
\end{equation*}

A $\vname{IBE}_{\bar{U}}$ scheme $\mathcal{E}$ is said to be attribute-hiding if
for all pairs of PPT algorithms $\mathcal{A} := (\mathcal{A}_1,
\mathcal{A}_2)$, it holds that
$\adv{\mathcal{E}}{\aname{{\bar{U}Priv}}}(\mathcal{A}) \leq \fname{negl}(\lambda)$. Note that the above definition assumes adaptive adversaries, but can be easily modified to accommodate the
non-adaptive case.

\subsection{Universal Anonymizers}
We now present an abstraction called a \emph{universal anonymizer}. With its help, we can
transform a universally-homomorphic, non-attribute-hiding IBE scheme
$\mathcal{E}$ into a non-universally homomorphic, attribute-hiding
scheme $\mathcal{E}^\prime$. In accordance with the property of
universal anonymity proposed in \cite{Hayashi:2005}, any party can
anonymize a given ciphertext.

Let $\mathcal{E} := (\aname{Setup}, \aname{KeyGen}, \aname{Encrypt},
\aname{Decrypt}, \aname{Eval})$ be a PE scheme parameterized with message space
$M$, attribute space $A$, class of predicates $\mathcal{F}$, and
class of circuits $\mathbb{C}$. Denote its ciphertext space by $\mathcal{C}$. Note that this definition of a universal anonymizer only suffices for simple functionalities such as IBE.

\begin{definition}\label{def:ua}
A \emph{universal anonymizer} $U_{\mathcal{E}}$ for a PE scheme 
$\mathcal{E}$ is a tuple $(\mathcal{G}, \mathcal{B},
\mathcal{B}^{-1}, Q_A, Q_{\mathcal{F}})$ where $\mathcal{G}$ is a
deterministic algorithm, $\mathcal{B}$ and $\mathcal{B}^{-1}$ are
randomized algorithms, and $Q_A$ and $Q_{\mathcal{F}}$ are efficient
maps, defined as follows:

\begin{itemize}
\item
$\mathbf{\aname{\mathcal{G}}(\vname{PP})}$:

On input the public parameter of an instance of $\mathcal{E}$, output
a parameters structure $\vname{params}$. This contains a description of
a modified ciphertext space $\hat{\mathcal{C}}$ as well as an integer
$d = \fname{poly}(\lambda)$ indicating the length of binary strings representing an attribute
class.

\item
$\mathbf{\aname{\mathcal{B}}(\vname{params}, \vec{c})}$:

On input parameters $\vname{params}$ and a ciphertext $\vec{c} \in
\mathcal{C}$, output an element of 
$\hat{\mathcal{C}}$.

\item
$\mathbf{\aname{\mathcal{B}}^{-1}(\vname{params}, \alpha, \hat{\vec{c})}}$:

On input parameters $\vname{params}$, a binary string $\alpha \in \{0, 1\}^d$ and an element of $\hat{\mathcal{C}}$, output an
element of $\mathcal{C}$.
\item
Both maps $Q_A$ and $Q_{\mathcal{F}}$ are indexed by $\vname{params}$:
${Q_A}_{\vname{params}} : A \to \{0, 1\}^d$ and ${Q_{\mathcal{F}}}_{\vname{params}} :
\mathcal{F} \to \{0, 1\}^d$.
\end{itemize}
\end{definition}
Note: $\vname{params}$ can be assumed to be an implicit input; it will
not be explicitly specified to simplify notation.

The binary string $\alpha$ is computed by means
of a map $Q_A : A \to \{0, 1\}^d$. In order for a decryptor to invert
$\mathcal{B}$, $\alpha$ must also be computable from any predicate that is satisfied by
an attribute that maps onto $\alpha$. Therefore, the map
$Q_{\mathcal{F}} : \mathcal{F} \to \{0, 1\}^d$ has the property that
for all $a \in A$ and $f \in \mathcal{F}$: 
\begin{equation*}
f(a) = 1 \implies Q_A(a) = Q_{\mathcal{F}}(f).
\end{equation*}
We define an equivalence relation $\sim$ on
$\mathcal{F}$ given by
\begin{equation*}
f_1 \sim f_2 \defeq Q_{\mathcal{F}}(f_1) = Q_{\mathcal{F}}(f_2).
\end{equation*} We have that 
\begin{equation*}
f \sim g  \iff \exists h_1, \hdots, h_k \in \mathcal{F}\quad \fname{supp}(f) \cap \fname{supp}(h_1) \neq \emptyset \land \hdots \land \fname{supp}(h_k) \cap \fname{supp}(g) \neq \emptyset.
\end{equation*}
It follows that each $\alpha$ is a
representative of an equivalence class in $\mathcal{F} /
\sim$. As a result, as mentioned earlier, our definition of a universal anonymizer above is only meaningful for ``simple'' functionalities such as IBE. For example, $|\mathcal{F} / \sim| = |\mathcal{F}|$ for an IBE scheme whose ciphertexts leak the recipient's identity.

Let $c$ be a ciphertext associated with an attribute $a$. Let $\alpha =
Q_A(a)$. Informally, 
$c^\prime := \mathcal{B}^{-1}(\alpha, \mathcal{B}(c))$ should ``behave'' like $c$; that is, (1)
it
should have the same homomorphic ``capacity'' and (2) decryption with a secret
key for any $f$ should have the same output as that for $c$.
A stronger requirement captured in our formal correctness criterion defined Appendix \ref{ap:unianon} is that $c$ and $c^\prime$ should be indistinguishable even when a distinguisher is given
access to $\vname{MSK}$.

A universal anonymizer is employed in the following
generic transformation from a universally-homomorphic, non-attribute-hiding IBE scheme
$\mathcal{E}$ to a non-universally homomorphic, attribute-hiding
scheme $\mathcal{E}^\prime$. 

The transformation is achieved by setting:
\begin{itemize}
\item
$
\aname{\mathcal{E}^\prime.Encrypt}(\vname{PP}, a, m) :=$
\begin{algorithmic}
\State
$\mathcal{B}(\aname{\mathcal{E}.Encrypt}(\vname{PP}, a, m))
$
\end{algorithmic}
\item
$
\aname{\mathcal{E}^\prime.Decrypt}(\vname{SK}_f, c) :=$
\begin{algorithmic}
\State
$\aname{\mathcal{E}.Decrypt}(\vname{SK}_f,
\mathcal{B}^{-1}(Q_{\mathcal{F}}(f), c))$
\end{algorithmic}
\item
$\aname{\mathcal{E}^\prime.Eval}(\vname{PP}, \alpha, C, c_1, \hdots,
c_\ell) := $

\begin{algorithmic}
\State
\Return $\mathcal{B}(\aname{\mathcal{E}.Eval}(\vname{PP}, C, \mathcal{B}^{-1}(\alpha, c_1), \hdots,
\mathcal{B}^{-1}(\alpha, c_\ell)))$
\end{algorithmic}

\end{itemize}
Denote the above transformation by
$T_{U_{\mathcal{E}}}(\mathcal{E})$. We leave to future work the task
  of establishing (generic) sufficient conditions that $\mathcal{E}$
  must satisfy to ensure that  $\mathcal{E}^\prime :=
T_{U_{\mathcal{E}}}(\mathcal{E})$ is an attribute-hiding
$\vname{HPE}_{\bar{U}}$ scheme.

An instantiation of a universal anonymizer for our XOR homomorphic
scheme is given in Appendix \ref{ap:unianon_inst}.

\subsection{Applications (Brief Overview)}\label{sec:app}
It turns out that XOR-homomorphic cryptosystems have been considered
to play an important part in
several applications.  The most well-known and widely-used \emph{unbounded} XOR-homomorphic
public-key cryptosystem is Goldwasser-Micali (GM) \cite{GM82}, which
is based on the quadratic residuosity problem. Besides being used in protocols such as private information retrieval (PIR), GM has been employed in some specific applications such as:
\begin{itemize}
\item
Peng, Boyd and Dawson (PBD) \cite{Peng:2005} propose a sealed-bid auction
system that makes extensive use of the GM cryptosystem.
\item
Bringer et al. \cite{Bringer:2007} apply GM to biometric authentication. It is used in two
primary ways; (1) to achieve PIR and (2) to assist in computing the hamming
distance between a recorded biometric template and a reference
one.
\end{itemize}
Perhaps in some of these applications, a group-homomorphic identity-based scheme may be of import, although the authors concede that no specific usage scenario has been identified so far.

With regard to performance, our construction requires $8$
multiplications in $\ZZ_N$ for a single homomorphic operation in
comparison to a single multiplication in
GM. Furthermore, the construction has higher ciphertext expansion than
GM by a factor of $4$. Encryption involves 2 modular inverses and 6
multiplications (only 4 if the strongly homomorphic property is
forfeited). In comparison, GM only requires $1.5$ multiplications on
average.

\section{Conclusions and Future Work}\label{sec:conclusions}
We have presented a characterization of homomorphic encryption in the
PE setting and classified schemes based on the properties of their
attribute homomorphisms. Instantiations of certain homomorphic properties were presented for
inner-product PE. However, it is clear that meaningful attribute homomorphisms
are limited. We leave to future work the exploration of homomorphic encryption with
access policies in a more general setting .

In this paper, we introduced a new XOR-homomorphic variant of
the Cocks' IBE scheme and showed that it is strongly
homomorphic. However, we failed to fully preserve the homomorphic
property in anonymous variants; that is, we could not construct an
anonymous universally-homomorphic variant. We leave this as an open
problem. As a compromise, however, a weaker primitive (non-universal IBE) was introduced
along with a related security notion. Furthermore, a transformation
strategy adapted from the work of Ateniese and
Gasti \cite{AG09} was exploited to obtain anonymity for our
XOR-homomorphic construction in this
weaker primitive. 

In future work, it is hoped to construct other group homomorphic IBE
schemes, and possibly for more general classes of predicates than the
IBE functionality.

Noteworthy problems, which we believe are still open:
\begin{enumerate}
\item
Somewhat-homomorphic IBE scheme (even non-adaptive security in the
ROM)
\item
(Unbounded) Group homomorphic IBE schemes for $(\ZZ_m, +)$ where $m =
O(2^\lambda)$ and $(\ZZ^\ast_p, \ast)$ for prime $p$.
Extensions include anonymity and support for a wider class of
predicates beyond the IBE functionality.
\end{enumerate}

\bibliographystyle{splncs}
\small{\bibliography{../bibliography/main}}

\ifx \undefined \cprime \def \cprime {$\mathsurround=0pt '$}\fi\ifx \undefined
  \Dbar \def \Dbar {\leavevmode\raise0.2ex\hbox{--}\kern-0.5emD} \fi\ifx
  \undefined \mathbb \def \mathbb #1{{\bf #1}}\fi\ifx \undefined \mathrm \def
  \mathrm #1{{\rm #1}}\fi\ifx \undefined \operatorname \def \operatorname
  #1{{\rm #1}}\fi\hyphenation{ Aba-di Arch-ives Ding-yi for-ge-ry
  Go-pa-la-krish-nan Hi-de-ki Kraw-czyk Lands-verk Law-rence Leigh-ton Mich-ael
  Moell-er North-ridge para-digm para-digms Piep-rzyk Piv-e-teau Ram-kilde
  Re-tro-fit-ting Rich-ard Sho-stak Si-ro-mo-n-ey Ste-ph-en The-o-dore Tho-m-as
  Tzone-lih venge-ance Will-iam Ye-sh-i-va }
\begin{thebibliography}{10}

\bibitem{Gentry2009}
Gentry, C.:
\newblock Fully homomorphic encryption using ideal lattices.
\newblock Proceedings of the 41st annual ACM Symposium on Theory of Computing
  STOC 09 (2009)  169

\bibitem{Smart10}
Smart, N., Vercauteren, F.:
\newblock Fully homomorphic encryption with relatively small key and ciphertext
  sizes.
\newblock In Nguyen, P., Pointcheval, D., eds.: Public Key Cryptography -- PKC
  2010. Volume 6056 of Lecture Notes in Computer Science.
\newblock Springer Berlin / Heidelberg (2010)  420--443

\bibitem{Dijk10}
van Dijk, M., Gentry, C., Halevi, S., Vaikuntanathan, V.:
\newblock Fully homomorphic encryption over the integers.
\newblock In Gilbert, H., ed.: Advances in Cryptology -- EUROCRYPT 2010. Volume
  6110 of Lecture Notes in Computer Science.
\newblock Springer Berlin / Heidelberg (2010)  24--43

\bibitem{BV11a}
Brakerski, Z., Vaikuntanathan, V.:
\newblock {Fully Homomorphic Encryption from Ring-LWE and Security for Key
  Dependent Messages, Advances in Cryptology -- CRYPTO 2011}.
\newblock Volume 6841 of Lecture Notes in Computer Science.
\newblock Springer Berlin / Heidelberg, Berlin, Heidelberg (2011)  505--524

\bibitem{Brakerski2011b}
Brakerski, Z., Vaikuntanathan, V.:
\newblock {Efficient Fully Homomorphic Encryption from (Standard) LWE}.
\newblock Cryptology ePrint Archive, Report 2011/344 (2011)
  \url{http://eprint.iacr.org/}.

\bibitem{GM82}
Goldwasser, S., Micali, S.:
\newblock Probabilistic encryption \& how to play mental poker keeping secret
  all partial information.
\newblock In: Proceedings of the fourteenth annual ACM symposium on Theory of
  computing. STOC '82, New York, NY, USA, ACM (1982)  365--377

\bibitem{Paillier:1999:PKC}
Paillier, P.:
\newblock Public-key cryptosystems based on composite degree residuosity
  classes.
\newblock In Stern, J., ed.: EUROCRYPT. Volume 1592 of Lecture Notes in
  Computer Science., Springer (1999)  223--238

\bibitem{ElGamal:1985:PKCa}
ElGamal, T.:
\newblock A public key cryptosystem and a signature scheme based on discrete
  logarithms.
\newblock IEEE Transactions on Information Theory \textbf{31} (1985)  469--472

\bibitem{Katz:2008}
Katz, J., Sahai, A., Waters, B.:
\newblock Predicate encryption supporting disjunctions, polynomial equations,
  and inner products.
\newblock In: Proceedings of the theory and applications of cryptographic
  techniques 27th annual international conference on Advances in cryptology.
  EUROCRYPT'08, Berlin, Heidelberg, Springer-Verlag (2008)  146--162

\bibitem{Boneh11}
Boneh, D., Sahai, A., Waters, B.:
\newblock Functional encryption: Definitions and challenges.
\newblock In Ishai, Y., ed.: Theory of Cryptography. Volume 6597 of Lecture
  Notes in Computer Science.
\newblock Springer Berlin / Heidelberg (2011)  253--273

\bibitem{Gennaro:2010}
Gennaro, R., Gentry, C., Parno, B.:
\newblock Non-interactive verifiable computing: outsourcing computation to
  untrusted workers.
\newblock In: Proceedings of the 30th annual conference on Advances in
  Cryptology. CRYPTO'10, Berlin, Heidelberg, Springer-Verlag (2010)  465--482

\bibitem{Lopez-Alt:2012}
L\'{o}pez-Alt, A., Tromer, E., Vaikuntanathan, V.:
\newblock On-the-fly multiparty computation on the cloud via multikey fully
  homomorphic encryption.
\newblock In: Proceedings of the 44th symposium on Theory of Computing. STOC
  '12, New York, NY, USA, ACM (2012)  1219--1234

\bibitem{NaccacheTalk10}
Naccache, D.:
\newblock Is theoretical cryptography any good in practice? (2010) Talk given
  at CHES 2010 and Crypto 2010.

\bibitem{Brakerski:2011_v1}
Brakerski, Z., Vaikuntanathan, V.:
\newblock {Efficient Fully Homomorphic Encryption from (Standard) LWE}.
\newblock Cryptology ePrint Archive, Report 2011/344 Version: 20110627:080002
  (2011) \url{http://eprint.iacr.org/}.

\bibitem{GPV}
Gentry, C., Peikert, C., Vaikuntanathan, V.:
\newblock Trapdoors for hard lattices and new cryptographic constructions.
\newblock In: STOC '08: Proceedings of the 40th annual ACM symposium on Theory
  of computing, New York, NY, USA, ACM (2008)  197--206

\bibitem{Gentry:2010}
Gentry, C., Halevi, S., Vaikuntanathan, V.:
\newblock {A Simple BGN-Type Cryptosystem from LWE}.
\newblock In Gilbert, H., ed.: EUROCRYPT. Volume 6110 of Lecture Notes in
  Computer Science., Springer (2010)  506--522

\bibitem{BGN05}
Boneh, D., Goh, E.J., Nissim, K.:
\newblock Evaluating 2-{DNF} formulas on ciphertexts.
\newblock In Kilian, J., ed.: TCC. Volume 3378 of Lecture Notes in Computer
  Science., Springer (2005)  325--341

\bibitem{Benaloh:1994}
Benaloh, J.:
\newblock Dense probabilistic encryption.
\newblock In: In Proceedings of the Workshop on Selected Areas of Cryptography.
  (1994)  120--128

\bibitem{Galbraith:2002}
Galbraith, S.D.:
\newblock {Elliptic Curve Paillier Schemes}.
\newblock J. Cryptology \textbf{15} (2002)  129--138

\bibitem{Golle:2002}
Golle, P., Jakobsson, M., Juels, A., Syverson, P.:
\newblock Universal re-encryption for mixnets.
\newblock In: IN PROCEEDINGS OF THE 2004 RSA CONFERENCE, CRYPTOGRAPHER'S TRACK,
  Springer-Verlag (2002)  163--178

\bibitem{Gjosteen:2005}
Gj{\o}steen, K.:
\newblock Homomorphic cryptosystems based on subgroup membership problems.
\newblock In: Proceedings of the 1st international conference on Progress in
  Cryptology in Malaysia. Mycrypt'05, Berlin, Heidelberg, Springer-Verlag
  (2005)  314--327

\bibitem{Armknecht:2012}
Armknecht, F., Katzenbeisser, S., Peter, A.:
\newblock Group homomorphic encryption: characterizations, impossibility
  results, and applications.
\newblock Designs, Codes and Cryptography (2012)  1--24

\bibitem{Bellare01}
Bellare, M., Boldyreva, A., Desai, A., Pointcheval, D.:
\newblock Key-privacy in public-key encryption, Springer-Verlag (2001)
  566--582

\bibitem{Hayashi:2005}
Hayashi, R., Tanaka, K.:
\newblock Universally anonymizable public-key encryption.
\newblock In Roy, B.K., ed.: ASIACRYPT. Volume 3788 of Lecture Notes in
  Computer Science., Springer (2005)  293--312

\bibitem{AG09}
Ateniese, G., Gasti, P.:
\newblock Universally anonymous {IBE} based on the quadratic residuosity
  assumption.
\newblock In: Proceedings of the The Cryptographers' Track at the RSA
  Conference 2009 on Topics in Cryptology. CT-RSA '09, Berlin, Heidelberg,
  Springer-Verlag (2009)  32--47

\bibitem{Prabhakaran:2008}
Prabhakaran, M., Rosulek, M.:
\newblock Homomorphic encryption with cca security.
\newblock In: Proceedings of the 35th international colloquium on Automata,
  Languages and Programming, Part II. ICALP '08, Berlin, Heidelberg,
  Springer-Verlag (2008)  667--678

\bibitem{COCKS01}
Cocks, C.:
\newblock An identity based encryption scheme based on quadratic residues.
\newblock In: Proceedings of the 8th IMA International Conference on
  Cryptography and Coding, London, UK, Springer-Verlag (2001)  360--363

\bibitem{BGH07}
Boneh, D., Gentry, C., Hamburg, M.:
\newblock Space-efficient identity based encryption without pairings.
\newblock In: FOCS, IEEE Computer Society (2007)  647--657

\bibitem{GoldwasserHELecture11}
Goldwasser, S.:
\newblock Lecture: Introduction to homomorophic encryption (2011)
  \url{http://www.cs.bu.edu/~reyzin/teaching/s11cs937/notes-shafi-1.pdf}. Last
  Checked on March 31st 2013.

\bibitem{AGVW12ePrint}
Agrawal, S., Gorbunov, S., Vaikuntanathan, V., Wee, H.:
\newblock Functional encryption: New perspectives and lower bounds.
\newblock Cryptology ePrint Archive, Report 2012/468 (2012)
  \url{http://eprint.iacr.org/}.

\bibitem{BO12ePrint}
Bellare, M., O'Neill, A.:
\newblock Semantically-secure functional encryption: Possibility results,
  impossibility results and the quest for a general definition.
\newblock Cryptology ePrint Archive, Report 2012/515 (2012)
  \url{http://eprint.iacr.org/}.

\bibitem{Wei:2009}
Wei, R., Ye, D.:
\newblock Delegate predicate encryption and its application to anonymous
  authentication.
\newblock In: Proceedings of the 4th International Symposium on Information,
  Computer, and Communications Security. ASIACCS '09, New York, NY, USA, ACM
  (2009)  372--375

\bibitem{Rothblum:2011}
Rothblum, R.:
\newblock Homomorphic encryption: From private-key to public-key.
\newblock In Ishai, Y., ed.: TCC. Volume 6597 of Lecture Notes in Computer
  Science., Springer (2011)  219--234

\bibitem{Boneh:2007}
Boneh, D., Waters, B.:
\newblock Conjunctive, subset, and range queries on encrypted data.
\newblock In: Proceedings of the 4th conference on Theory of Cryptography.
  TCC'07, Berlin, Heidelberg, Springer-Verlag (2007)  535--554

\bibitem{Alzaid:2008}
Alzaid, H., Foo, E., Nieto, J.G.:
\newblock Secure data aggregation in wireless sensor network: a survey.
\newblock In: Proceedings of the sixth Australasian conference on Information
  security - Volume 81. AISC '08, Darlinghurst, Australia, Australia,
  Australian Computer Society, Inc. (2008)  93--105

\bibitem{Agrawal11}
Agrawal, S., Freeman, D.M., Vaikuntanathan, V.:
\newblock Functional encryption for inner product predicates from learning with
  errors.
\newblock In Lee, D.H., Wang, X., eds.: ASIACRYPT. Volume 7073 of Lecture Notes
  in Computer Science., Springer (2011)  21--40

\bibitem{Boneh:2004}
Boneh, D., Crescenzo, G.D., Ostrovsky, R., Persiano, G.:
\newblock Public key encryption with keyword search.
\newblock In Cachin, C., Camenisch, J., eds.: EUROCRYPT. Volume 3027 of Lecture
  Notes in Computer Science., Springer (2004)  506--522

\bibitem{Armknecht:2010}
Armknecht, F., Katzenbeisser, S., Peter, A.:
\newblock Group homomorphic encryption: Characterizations, impossibility
  results, and applications.
\newblock Cryptology ePrint Archive, Report 2010/501 (2010)
  \url{http://eprint.iacr.org/}.

\bibitem{CHT2013af}
Clear, M., Hughes, A., Tewari, H.:
\newblock Homomorphic encryption with access policies: Characterization and new
  constructions.
\newblock In Youssef, A., Nitaj, A., Hassanien, A., eds.: Progress in
  Cryptology – AFRICACRYPT 2013. Volume 7918 of Lecture Notes in Computer
  Science.
\newblock Springer Berlin Heidelberg (2013)  61--87

\bibitem{Peng:2005}
Peng, K., Boyd, C., Dawson, E.:
\newblock A multiplicative homomorphic sealed-bid auction based on
  goldwasser-micali encryption.
\newblock In Zhou, J., Lopez, J., Deng, R.H., Bao, F., eds.: ISC. Volume 3650
  of Lecture Notes in Computer Science., Springer (2005)  374--388

\bibitem{Bringer:2007}
Bringer, J., Chabanne, H., Izabach\`{e}ne, M., Pointcheval, D., Tang, Q.,
  Zimmer, S.:
\newblock An application of the {Goldwasser-Micali} cryptosystem to biometric
  authentication.
\newblock In: Proceedings of the 12th Australasian conference on Information
  security and privacy. ACISP'07, Berlin, Heidelberg, Springer-Verlag (2007)
  96--106

\end{thebibliography}

\appendix
\section{Group-Homomorphic Encryption Generalized for PE}\label{ap:ghe}
\begin{definition}[Extension of Definition \ref{def:ghom}]
\label{def:ghom_ext}
Let $\mathcal{E} = (G, K, E, D)$ be a PE scheme with message space
$M$, attribute space $A$, ciphertext space $\hat{\mathcal{C}}$ and
class of predicates $\mathcal{F}$.  The
scheme $\mathcal{E}$ is group homomorphic if for
every $(\vname{PP}, \vname{MSK}) \gets G(1^\lambda)$, every $f \in
\mathcal{F} : \fname{supp}(f) \neq \emptyset$, and every $\vname{sk}_f \gets
K(\vname{MSK}, f)$, the
message space $(M, \cdot)$ is a non-trivial group, and there is a binary operation
$\ast : \hat{\mathcal{C}}^2 \to
\hat{\mathcal{C}}$ such that the following properties are satisfied
for the restricted ciphertext space
$\hat{\mathcal{C}_f} = \{c \in
\hat{\mathcal{C}} : D_{\vname{sk}_f}(c) \neq \bot\}$:

\begin{enumerate}
\item
$(\hat{\mathcal{C}_f}, \ast)$ is a non-trivial group.
\item
The set of all encryptions $\mathcal{C}_f = \{c \in \hat{\mathcal{C}_f}  \mid c \gets E(\vname{PP}, a, m), a \in \fname{supp}(f), m \in M\}$ is a subgroup of $\hat{\mathcal{C}_f}$ with respect to the operation $\ast$.
\item
The restricted decryption $D_{\vname{sk}_f}^\ast := D_{\vname{sk}_f |
  \hat{\mathcal{C}_f}}$ is surjective and $\forall c, c^\prime \in
\hat{\mathcal{C}_f} \quad D_{\vname{sk}_f}(c \ast c^\prime) =
D_{\vname{sk}_f}(c) \cdot D_{\vname{sk}_f}(c^\prime)$.
\item
The following distributions are computationally indistinguishable:
\begin{equation*}
\{(\vname{PP}, f, \vname{sk}_f, S, c) \mid c
\samplefrom \mathcal{C}_{f}, S \subset \{\vname{sk}_g \gets K(g) : g \in \mathcal{F}\}\} \cind \{(\vname{PP}, f, \vname{sk}_f, S, \hat{c}) \mid \hat{c}
\samplefrom \hat{\mathcal{C}_f}, S \subset \{\vname{sk}_g \gets K(g) : g \in \mathcal{F}\}\}.
\end{equation*}
\item
There is an efficient function $\tau : A \to \mathcal{F}$ such that for any $a \in A$, $f = \tau(a)$ satisfies
\[
g(a^\prime) = g(a) \text{ for all } a^\prime \in \fname{supp}(f), g \in \mathcal{F}.
\]
\end{enumerate}
\end{definition}
Armknecht, Katzenbeisser and Peter give a characterization of group-homomorphic public key cryptosystems  \cite{Armknecht:2012}. Their characterization includes the condition that the secret key contain an efficient predicate, or \emph{decision function}, $\delta : \hat{\mathcal{C}} \to \{0, 1\}$ satisfying
\[
\delta(c) = 1 \iff c \in \mathcal{C}
\]
where $\hat{\mathcal{C}}$ denotes the ciphertext space and $\mathcal{C} \subseteq \hat{\mathcal{C}}$ denotes the set of legally-generated ciphertexts under the public key (i.e. the image of the encryption algorithm over all messages and random coins). Now generalizing this to PE in the above definition yields a decision function $\delta_f : \hat{\mathcal{C}} \to \{0, 1\}$ with
\[
\delta_f(c) = 1 \iff c \in \mathcal{C}_f.
\]
We can show that such a decision function does not always exist. A counterexample is our XOR-homomorphic IBE system from Section \ref{sec:cocks}. Let $a = H(\vname{id})$ and for some identity $\vname{id}$. Let $f : A \to \{0, 1\}$ be the point function that is nonzero at exactly $\vname{id} \in A$. Then $\mathcal{C}_f$ corresponds to $\{(c(x), d(x), a) : c(x) \in S_a, d(x) \in S_{-a}\}$ and $\hat{\mathcal{C}}_f$ corresponds to $\{(c(x), d(x), a) : c(x) \in G_a, d(x) \in G_{-a}\}$. However, there is no efficient distinguisher that can distinguish between $S_a$ and $G_a$ (or $S_{-a}$ and $G_{-a}$) without access to the factorization of $N$ (i.e. the master secret key). It follows that there is no efficient decision function. This necessitates property 3 in the above definition in order to extend the abstract characterizations of IND-CCA1 security in \cite{Armknecht:2012} to the PE setting. Because of property 3, it suffices to define $\delta_f$ as
\begin{equation}\label{eq:delta_f}
\delta_f(c) = 1 \iff c \in \hat{\mathcal{C}}_f.
\end{equation}

We also extend the notion of GIFT (Generic shIFt-Type) from Definition 3 in \cite{Armknecht:2012}. We defer the reader to this paper for a formal definition of GIFT. Informally, a GIFT PE scheme satisfies the following:
\begin{itemize}
\item
The public parameters $\vname{PP}$ contains information to determine a non-trivial, proper normal subgroup $\mathcal{N}_f$ for every group $\mathcal{C}_f$.
\item
It holds that for every $f, g \in \mathcal{F}$, the systems of representatives $\mathcal{R}_f = \mathcal{C}_f / \mathcal{N}_f$ and $\mathcal{R}_f = \mathcal{C}_g / \mathcal{N}_g$ have the same cardinality; that is, $|\mathcal{R}_f| = |\mathcal{R}_g|$.
\item
$\vname{PP}$ contains an efficient function $\psi : \mathcal{F} \times M \to \hat{\mathcal{C}}$ with the property that $\psi_f = \psi(f, \cdot)$ for any $f \in \mathcal{F}$ is an isomorphism between $M$ and $\mathcal{R}_f$.
\item
To encrypt a message $m \in M$ under attribute $a \in A$, an encryptor:
\begin{enumerate}
\item
computes $f^\prime \gets \tau(a)$,
\item
chooses a random $n \samplefrom \mathcal{N}_{f^\prime}$,
\item
and outputs the ciphertext $\psi_{f^\prime}(m) \ast n \in \mathcal{C}_{f^\prime}$.
\end{enumerate}
\item
A secret key $\vname{sk}_f$ for predicate $f \in \mathcal{F}$ contains an efficient description of $\psi_f^{-1} \circ \mu_f$ where $\mu_f : \hat{\mathcal{C}_f} \to \mathcal{R}_f$ such that $r = \mu(c)$ is the unique representative with $c = r \ast n$ where $n \in \mathcal{N}_f$.
\end{itemize}

\subsection{Interactive Splitting Oracle-Assisted Subgroup Membership Problem (ISOAP)}
Let $G$ be a PPT algorithm that takes as input a security parameter $\lambda$ and outputs a tuple $(\hat{\mathcal{G}}, \mathcal{I}, \mathbb{G}, k)$ where $\hat{\mathcal{G}}$ is a finite semigroup, $\mathcal{I}$ is a set of indices and $\mathbb{G}$ and $k$ are defined momentarily. Firstly, $\mathbb{G}$ is a family $\{(\mathcal{G}_i, \mathcal{N}_i, \mathcal{R}_i)\}_{i \in \mathcal{I}}$ where $\mathcal{G}_i \subseteq \hat{\mathcal{G}}$ is a non-trivial group, $\mathcal{N}_i$ is a proper, non-trivial subgroup of $\mathcal{G}_i$ and $\mathcal{R}_i \subset \mathcal{G}_i$ is a finite set of representatives of $\mathcal{G}_i / \mathcal{N}_i$. It is required that $|R_i| = |R_j|$ for all $i, j \in \mathcal{I}$. Finally, $k$ is efficient trapdoor information that allows us to efficiently solve the splitting problem (SP) in any group $\mathcal{G}_i$; that is, given some $c \in \mathcal{G}_i$, the goal of SP is to find the unique $r \in \mathcal{R}_i$ and $n \in \mathcal{N}_i$ such that $c = r \ast n$. We let $K$ be a PPT algorithm that uses $k$ and takes an index $i \in \mathcal{I}$ as input, and outputs a description of an efficient function $\sigma_i : \mathcal{G}_i \to \mathcal{R}_i \times \mathcal{N}_i$. Such a function solves SP in $\mathcal{G}_i$.For brevity, we set $K^\prime := K_k$.

We define an interactive version of the problem $\aname{SOAP}$ from \cite{Armknecht:2012}, which we refer to as $\aname{ISOAP}$. This is a subgroup membership problem relating to a group chosen by the adversary who in addition is granted access to a ``splitting oracle'' for that group.

The game that defines $\aname{ISOAP}$ proceeds as follows. Prior to the challenge phase, the adversary is granted access to a ``splitting oracle'' $\mathcal{O}_{\vname{SP}}^{\hat{\mathcal{G}}, \mathcal{I}, \mathbb{G}, K}$ that takes an index $i \in \mathcal{I}$ and an element $c$ of $\hat{\mathcal{G}}$, and answers with $\bot$ if $c \notin \mathcal{G}_i$; otherwise, it answers with $\sigma_i(c)$. In addition, the adversary is given access to another oracle $\mathcal{O}_{K_1}^{\hat{\mathcal{G}}, \mathcal{I}, \mathbb{G}, K, Q}$ in the first phase which responds to a query for an index $i \in \mathcal{I}$ by storing $i$ in a cache $\mathcal{Q}$ and returning $K(i)$.

Then the adversary chooses a ``challenge'' group by specifying an index $\vname{ind} \in \mathcal{I}$ subject to the condition that $\mathcal{G}_{\vname{ind}} \cap \mathcal{G}_j = \emptyset$ for every $j \in Q$.  It receives a challenge element $c^\ast \in \mathcal{G}_{\vname{ind}}$.

In the second phase, the adversary is given access to a more restricted oracle $\mathcal{O}_{K_2}^{\hat{\mathcal{G}}, \mathcal{I}, \mathbb{G}, K, c^\ast}$ that when queried on index $i \in \mathcal{I}$, returns $K(i)$ if $c^\ast \notin \mathcal{G}_i$, and returns $\bot$ otherwise.
\paragraph
\noindent Experiment $\exper{(\mathcal{A}_1, \mathcal{A}_2), G, K}{\aname{ISOAP}}(\lambda)$:
\begin{enumerate}
\item
$(\hat{\mathcal{G}}, \mathcal{I}, \mathbb{G}, k) \gets G(\lambda)$. $K^\prime := K_k$.
\item
$s, \vname{ind} \gets \mathcal{A}_1^{\mathcal{O}_{\vname{SP}}^{\hat{\mathcal{G}}, \mathcal{I}, \mathbb{G}, K^\prime}(\cdot, \cdot), \mathcal{O}_{K_1}^{\hat{\mathcal{G}}, \mathcal{I}, \mathbb{G}, K^\prime, Q}}(\hat{\mathcal{G}}, \mathcal{I}, \mathbb{G})$.
\item
Choose $b \samplefrom \{0, 1\}$. If $b = 1$: $c^\ast \samplefrom \mathcal{G}_{\vname{ind}}$. Otherwise, $c^\ast \samplefrom \mathcal{N}_{\vname{ind}}$.
\item
$b^\prime \gets \mathcal{A}_2^{\mathcal{O}_{K_2}^{\hat{\mathcal{G}}, \mathcal{I}, \mathbb{G}, K^\prime, c^\ast}}(\hat{\mathcal{G}}, \mathcal{I}, \mathbb{G}, s, c^\ast)$
\item
Output $1$ if $b^\prime = b$. Output $0$ otherwise.
\end{enumerate}

\begin{theorem} Let $\mathcal{E} = (G, K, E, D)$ be a GIFT PE scheme. Then $\mathcal{E}$ is IND-AD-CCA1 secure if and only if $\aname{ISOAP}$ is hard relative to an algorithm $G^\prime$ that derives the tuple it outputs, namely $(\hat{\mathcal{C}}, \mathcal{F}, \mathbb{G} := \{(\hat{\mathcal{C}_f}. \mathcal{R}_f, \hat{\mathcal{N}_f)}\}_{f \in \mathcal{F}}, \vname{MSK})$, from $(\vname{PP}, \vname{MSK}) \gets G(\lambda)$ where $\hat{\mathcal{C}_f} = \mathcal{R}_f \ast \hat{\mathcal{N}_f}$ for every $f \in \mathcal{F}$.
\end{theorem}
\begin{proof}[sketch]
The proof is similar to the proof of Theorem 3 in \cite{Armknecht:2012}.

Firstly, we show that the hardness of $\aname{ISOAP}$ implies the IND-AD-CCA1 security of $\mathcal{E}$. Suppose $\aname{ISOAP}$ is hard. Assume that $\mathcal{E}$ is not IND-AD-CCA1 secure. Then there is an algorithm $\mathcal{B}$ that has a non-negligible advantage $\epsilon$ attacking the IND-AD-CCA1 security of $\mathcal{E}$. This algorithm can be used to construct an adversary $\mathcal{A}^{\vname{ISOAP}} = (\mathcal{A}^{\vname{ISOAP}}_1, \mathcal{A}^{\vname{ISOAP}}_2)$ that obtains an advantage of $\frac{1}{2}\epsilon$ against $\aname{ISOAP}$.  Now $\mathcal{A}^{\vname{ISOAP}}_1$ can simulate $\vname{PP}$ and forward it to $\mathcal{B}$. It handles a secret key query for $f \in \mathcal{F}$ by querying its oracle $\mathcal{O}_{K_1}$. Furthermore, it handles a decryption query for $(f, c)$ where $f \in \mathcal{F}$ and $c \in \hat{\mathcal{C}}$ by returning $\bot$ if $\delta_f(c) = 0$ (see the definition of $\delta_f$ in Equation \ref{eq:delta_f}) and responding with $\psi_f^{-1}(r)$ otherwise, where $(r, n) \gets \mathcal{O}_\vname{SP}(f, c)$. When $\mathcal{B}$ chooses a target attribute $a^\ast$ and two messages $(m_0, m_1)$, $\mathcal{A}^{\vname{ISOAP}}_1$ computes its target index $f^\ast = \tau(a^\ast)$ and forwards it to the $\aname{ISOAP}$ challenger. $\mathcal{A}^{\vname{ISOAP}}_2$ derives an IND-AD-CCA1 challenge ciphertext from its $\aname{ISOAP}$ challenge $c^\ast$ by choosing a random bit $t \samplefrom \{0, 1\}$ and computing $c^\prime = \psi_{f^\ast}(m_t) \ast c^\ast$. It hands $c^\prime$ to $\mathcal{B}$. $\mathcal{A}^{\vname{ISOAP}}_2$ responds to secret key queries by using its oracle $\mathcal{O}_{K_2}$ in a similar manner to $\mathcal{A}^{\vname{ISOAP}}$.
 Finally, it outputs $t \oplus b^\prime$ where $b^\prime$ is $\mathcal{B}$'s guess. Now let $b$ be the bit chosen by the $\aname{ISOAP}$ challenger. If $b = 0$, then $c^\prime$ is indistinguishable from a correctly distributed encryption of $m_t$. It is indistinguishable due to property 3 in Definition \ref{ap:ghe} since $c^\prime \in \hat{\mathcal{C}}_{f^\ast}$ whereas a legally-generated ciphertext lies within $\mathcal{C}_{f^\ast}$. Denote $\mathcal{B}$'s advantage distinguishing both cases by $\adv{\mathcal{B}, \mathcal{E}}{\aname{IND-CT}}$. If $b = 1$, then $c^\prime$ is an encryption of a random element of $M$, which contains no information about $t$, forcing $\mathcal{B}$'s advantage to zero. Therefore, the overall advantage of $\mathcal{A}^{\vname{ISOAP}}$ is $\adv{\mathcal{B}, \mathcal{E}}{\aname{IND-CT}} + \frac{1}{2}\epsilon$.

Now we prove the reverse direction. Suppose that $\mathcal{E}$ is IND-AD-CCA1 secure. Assume for the purpose of contradiction that there is an adversary $\mathcal{A}^{\vname{ISOAP}} = (\mathcal{A}^{\vname{ISOAP}}_1, \mathcal{A}^{\vname{ISOAP}}_2)$ whose advantage is $\alpha$ against $\aname{ISOAP}$. We can use $\mathcal{A}^{\vname{ISOAP}}$ to construct an adversary $\mathcal{B}$ to attack the IND-AD-CCA1 security of $\mathcal{E}$. Firstly, $\mathcal{B}$ derives $\hat{\mathcal{C}}$, $\mathcal{F}$ and $\mathbb{G}$ from $\vname{PP}$ and passes them to $\mathcal{A}^{\vname{ISOAP}}_1$. It simulates $\mathcal{O}_{K_1}$ by forwarding a query for $f$ to its secret key oracle and responding with a description of $\sigma_f$ derived from the secret key $\vname{sk}_f$ it receives. It simulates a query to $\mathcal{O}_{\vname{SP}}$ for $(f, c)$ by (1). querying its decryption oracle for $c$ to obtain $m^\prime$; (2). computing $r \gets \psi_f(m^\prime)$ and $n \gets r^{-1} \ast c$; and (3). responding with $(r, n)$.

Let $f^\ast$ be the target index outputted by $\mathcal{A}^{\vname{ISOAP}}_1$. Subsequently, $\mathcal{B}$ chooses an attribute $a^\ast \in \fname{supp}(f^\ast)$ and forwards $a^\ast$ as its challenge attribute. Furthermore, it chooses messages $m_0, m_1 \samplefrom M$ and forwards them to its challenger who responds with a challenge ciphertext $c^\ast$. Next $\mathcal{B}$ computes $c^\prime \gets c^\ast \ast E(a^\ast, m_0)^{-1} \ast n^\prime$ where $n^\prime \samplefrom \mathcal{N}_{f^\ast}$ and hands $c^\prime$ to $\mathcal{A}^{\vname{ISOAP}}_2$. Let $b^\prime$ be the bit guessed by $\mathcal{A}^{\vname{ISOAP}}$. Then $\mathcal{B}$ outputs $b^\prime$ as its guess. Recall how a ciphertext is generated by the IND-AD-CCA1 challenger. Firstly, the challenger samples bit $t \samplefrom \{0, 1\}$. Then the challenger computes $f \gets \tau(a)$ and sets $c^\ast := \psi_f(m_t) \ast n$ where $n \samplefrom \mathcal{N}_f$.It follows by definition of $\tau$ that $\mathcal{N}_f \subseteq \mathcal{N}_{f^\ast}$. This immediately implies that $\mathcal{R}_f = \mathcal{R}_{f^\ast}$. If $t = 0$, then $c^\prime$ is distributed according to a uniformly random element from $\mathcal{N}_{f^\ast}$, which results in an advantage of $\frac{1}{2}\alpha$ provided that $\mathcal{A}^{\vname{ISOAP}}_2$ cannot distinguish between $\mathcal{C}_f$ and $\hat{\mathcal{C}_f}$ (property 3 of Definition \ref{ap:ghe}). If $t = 1$, then $c^\prime$ is a uniform in $\hat{\mathcal{C}}_f$, and the advantage of $\mathcal{B}$ in this case is also $\frac{1}{2}\alpha$ provided that $\mathcal{A}^{\vname{ISOAP}}_2$ cannot distinguish between $\mathcal{C}_f$ and $\hat{\mathcal{C}_f}$. The overall advantage is therefore $\adv{\mathcal{A}^{\vname{ISOAP}}_2, \mathcal{E}}{\aname{IND-CT}} + \alpha$.
\qed
\end{proof}

\section{Correctness Condition for a Universal Anonymizer}\label{ap:unianon}

Let $\mathcal{E}$ be a H(PE) scheme with public index, and let
$U_{\mathcal{E}} := (\mathcal{G}, \mathcal{B}, \mathcal{B}^{-1}, Q_A, Q_{\mathcal{F}})$ be a universal anonymizer for $\mathcal{E}$. Define the distributions
$D_1 := \{(\vname{PP}, \vname{MSK}, \vname{params}, c) \mid  (\vname{PP}, \vname{MSK}) \leftarrow
\aname{\mathcal{E}.Setup}(1^\lambda),  \vname{params} \gets
\mathcal{G}(\vname{PP}), c \samplefrom \mathcal{C}\}$ and
$D_2 := \{(\vname{PP}, \vname{MSK}, \vname{params}, c^\prime) \mid  (\vname{PP}, \vname{MSK}) \leftarrow
\aname{\mathcal{E}.Setup}(1^\lambda), \vname{params} \gets
\mathcal{G}(\vname{PP}) , c \samplefrom \mathcal{C}, c^\prime \gets
\mathcal{B}^{-1}(Q_A(\fname{attr}(c)), \mathcal{B}(c))\}$ where
$\fname{attr}(c)$ returns the attribute associated with $c$. The
correctness condition for a universal anonymizer $U_{\mathcal{E}}$ is
that $D_1 \cind D_2$ (computationally indistinguishability).

\section{Instantiation of a Universal Anonymizer for Main Construction}\label{ap:unianon_inst}
The techniques from \cite{AG09} can be employed to construct a
universal anonymizer for $\aname{xhIBE}$. In this paper, the
basic version of their construction is adapted.

Let $L(\lambda)$ be the maximum bit-length of identities in
$\aname{xhIBE}$. A universal anonymizer
$\vname{AG}_{\aname{xhIBE}} :=$
\[(\aname{\vname{AG}_{\aname{xhIBE}}.\mathcal{G}}, \aname{\vname{AG}_{\aname{xhIBE}}.\mathcal{B}},
\aname{\vname{AG}_{\aname{xhIBE}}.{\mathcal{B}^{-1}}}, Q_A := H, Q_f := f_{\vec{id}} \mapsto H(\vec{id}))\]
 for $\aname{xhIBE}$ based on
the techniques of Ateniese and
Gasti is given as follows:

Let $\vname{Geom}(p)$ be a geometric distribution with parameter $p$.
\begin{algorithm}
\begin{algorithmic}
\caption{$\aname{\vname{AG}_{\aname{xhIBE}}.\mathcal{G}}(\vname{PP})$}
\State
$m \gets \lambda$ \Comment $\lambda$ can be derived from $\vname{PP}$
\State
$\vname{params} := (m, \lg{N})$ \Comment (length of members
of $\hat{\mathcal{C}}$ is $2(m + 1)\cdot \lg{N}$ bits, length of $\alpha$)
\State
\Return $\vname{params}$
\end{algorithmic}
\end{algorithm}

\begin{algorithm}
\begin{algorithmic}
\caption{$\aname{\vname{AG}_{\aname{xhIBE}}.\mathcal{B}}(\vname{params}, \vec{\psi})$}
\State
Parse $\vname{params}$ as $(m, L)$
\State
Parse $\vec{\psi}$ as $(c(x), d(x), a)$
\State
$k_1, k_2 \samplefrom \vname{Geom}(\frac{1}{2})$
\State
$k_1 \gets \fname{min}(k_1, m)$.
\State
$k_2 \gets \fname{min}(k_2, m)$.
\State
$t(x), v(x) \samplefrom \ZZ_N[x]$
\State
$z_1(x) \gets c(x) + t(x)$
\State
$z_2(x) \gets d(x) + v(x)$
\For{$1 \leq i < k_1$}
\Repeat \State $t_i(x) \samplefrom \ZZ_N[x]$ \Until{$\fname{GT}(a, z_1(x) - t_i(x), N) =
  -1$}
\EndFor
\State
$t_{k_1} \gets t(x)$
\For{$1 \leq i < k_2$}
\Repeat \State $v_i(x) \samplefrom \ZZ_N[x]$ \Until{$\fname{GT}(-a, z_2(x) - v_i(x), N) =
  -1$}
\EndFor
\State
$v_{k_2} \gets v(x)$
\For{$k_1 < i \leq m$}
\State $t_i(x) \samplefrom \ZZ_N[x]$
\EndFor
\For{$k_2 < i \leq m$}
\State $v_i(x) \samplefrom \ZZ_N[x]$
\EndFor
\State
\Return $\vec{\hat{\psi}} := ((z_1(x), t_1(x), \hdots, t_m(x)), (z_2(x), v_1(x), \hdots,
v_m(x))) \in \ZZ_N[x]^{2m + 2}$
\end{algorithmic}
\end{algorithm}

\begin{algorithm}
\begin{algorithmic}
\caption{$\aname{\vname{AG}_{\aname{xhIBE}}.\mathcal{B}}^{-1}(\vname{params}, a, \vec{\hat{\psi}})$}
\State
Parse $\vname{params}$ as $(m, L)$
\State
Parse $\vec{\hat{\psi}}$ as $((z_1(x), t_1(x), \hdots, t_m(x)), (z_2(x), v_1(x), \hdots,
v_m(x)))$
\State
$i \gets 1$
\While {$\fname{GT}(a, t_i(x) - z_1(x), N) \neq 1$}
\State
$i \gets i + 1$
\EndWhile
\State
$c(x) \gets t_i(x) - z_1(x)$
\State
$i \gets 1$
\While {$\fname{GT}(-a, v_i(x) - z_2(x), N) \neq 1$}
\State
$i \gets i + 1$
\EndWhile
\State
$d(x) \gets v_i(x) - z_2(x)$
\State
\Return $(c(x), d(x), a)$
\end{algorithmic}
\end{algorithm}

Let the set of valid ciphertexts $\mathcal{C}$ be defined as $\{(c(x),
d(x), a) \in \ZZ_N[x]^2 \times \ZZ_N : c(x) \in G_a, d(x) \in G_{-a}\}$.
Then for any $(\vname{PP}, \vname{MSK}) \gets
\aname{xhIBE.Setup}(1^\lambda)$ and $\vname{params} \gets \aname{\vname{AG}_{\aname{xhIBE}}.\mathcal{G}}(\vname{PP})$:
the correctness condition in Appendix \ref{ap:unianon} is trivially
satisfied since $\forall \vec{\psi} := (c(x),
d(x), a) \in \mathcal{C}$
\begin{equation*}
\vec{\psi} = \aname{\vname{AG}_{\aname{xhIBE}}.\mathcal{B}}^{-1}(a, \aname{\vname{AG}_{\aname{xhIBE}}.\mathcal{B}}(\vec{\psi}))
\end{equation*}
We can apply the transformation
\begin{equation*}
\aname{xhIBE}^\prime \gets
T_{\vname{AG}_{\aname{xhIBE}}}(\aname{xhIBE})
\end{equation*}
 described in
  the last section to obtain a scheme $\aname{xhIBE}^\prime$. The
  scheme in \cite{AG09} is shown to satisfy a security definition
  (ANON-IND-ID-CPA) in the random oracle model that is stronger than the
  attribute-hiding definition for IBE in the random oracle model. It
  can be easily shown with the help of Corollary \ref{cor:dist}
  that $\aname{xhIBE}^\prime$ is an attribute-hiding
  $\vname{HPE}_{\bar{U}}$ scheme for the IBE functionality supporting the
  group homomorphism $(\ZZ_2, +)$.

\end{document}